\documentclass[sigconf]{acmart}
\usepackage{popets}

\usepackage{bm,nicefrac}
\usepackage{cancel}
\usepackage{multirow}

\usepackage{framed,url,caption,subcaption,graphicx}
\usepackage{xcolor,soul, fontawesome,algpseudocode,algorithm,xspace,multirow,ctable,longtable,colortbl, lscape, pifont}
\usepackage{cleveref}
\usepackage{booktabs, siunitx,comment}

\definecolor{light-gray}{gray}{0.9}
\definecolor{darkgreen}{rgb}{0,0.5,0}
\definecolor{light-blue}{rgb}{0,.7,1}
\definecolor{red}{rgb}{.7, 0, 0}
\renewcommand{\arraystretch}{1.5}

\usepackage{enumitem}
\setlist{noitemsep, topsep=0pt, partopsep=0pt}
\graphicspath{{./figures/} }
\crefformat{section}{\S#2#1#3}
\crefformat{subsection}{\S#2#1#3}
\crefformat{subsubsection}{\S#2#1#3}

\expandafter\def\expandafter\UrlBreaks\expandafter{\UrlBreaks
  \do\a\do\b\do\c\do\d\do\e\do\f\do\g\do\h\do\i\do\j%
  \do\k\do\l\do\m\do\n\do\o\do\p\do\q\do\r\do\s\do\t%
  \do\u\do\v\do\w\do\x\do\y\do\z\do\A\do\B\do\C\do\D%
  \do\E\do\F\do\G\do\H\do\I\do\J\do\K\do\L\do\M\do\N%
  \do\O\do\P\do\Q\do\R\do\S\do\T\do\U\do\V\do\W\do\X%
  \do\Y\do\Z}

\newcommand{\F}{Figure }

\newcommand{\T}{Table }

\renewcommand{\S}{Section }
\newcommand{\A}{Alg.}
\newcount\Comments  \Comments=1  
\newcommand{\kibitz}[2]{\ifnum\Comments=1\textcolor{#1}{#2}\fi}
\newcommand{\STOP}[1]  {\kibitz{red}   {[\textbf{STOP HERE} - compiled: \today]}}

\newcommand{\R}{\mathbb{R}}

\newcommand{\paratitle}[1]{\vspace{.05in} \noindent\textbf{#1}}

\usepackage{bbm}
\usepackage{diagbox}
\usepackage{setspace}
\usepackage{hyperref}
\usepackage[flushleft]{threeparttable}
\setlist[enumerate]{wide=\parindent}

\def\ToolX{\textsc{De-Harpo}\xspace}
\def\Obfuscator{\textit{obfuscator}\xspace}

\def\Denoiser{\textit{denoiser}\xspace}

\setcopyright{popets}
\copyrightyear{YYYY}

\acmYear{YYYY}
\acmVolume{YYYY}
\acmNumber{X}
\acmDOI{XXXXXXX.XXXXXXX}
\acmISBN{}
\acmConference{Proceedings on Privacy Enhancing Technologies}
\begin{document}

\title{A Utility-Preserving Obfuscation Approach for YouTube Recommendations}


\author{
Jiang Zhang$^1$ {} {} {} {} 
Hadi Askari$^2$ {} {} {} {} 
Konstantinos Psounis$^1$ {} {} {} {} 
Zubair Shafiq$^2$}
\affiliation {
$^1$University of Southern California {} {} {} {} 
$^2$University of California, Davis \country{}}

\email{
{jiangzha, kpsounis}@usc.edu,
{haskari, zubair}@ucdavis.edu
}





\renewcommand{\shortauthors}{Zhang et al.}

\begin{abstract}
Online content platforms optimize engagement by providing personalized recommendations to their users.
These recommendation systems track and profile users to predict relevant content a user is likely interested in. 
While the personalized recommendations provide utility to users, the tracking and profiling that enables them poses a privacy issue.
%
There is increasing interest in building privacy-enhancing obfuscation approaches that do not rely on cooperation from online content platforms.
However, existing obfuscation approaches primarily focus on enhancing privacy but at the same time they degrade the utility because obfuscation introduces unrelated recommendations.
We design and implement \ToolX, an obfuscation approach for YouTube's recommendation system that not only \textit{obfuscates} a user's video watch history to protect privacy but then also \textit{denoises} the video recommendations by YouTube to preserve their utility. 
In contrast to prior obfuscation approaches, \ToolX adds a \Denoiser that makes use of a ``secret'' input (i.e., a user's actual watch history) as well as information that is also available to the adversarial recommendation system (i.e., obfuscated watch history and corresponding ``noisy" recommendations).
Our large-scale evaluation of \ToolX shows that it outperforms the state-of-the-art by a factor of 2$\times$ in terms of preserving utility for the same level of privacy, while maintaining stealthiness and robustness to de-obfuscation.
\end{abstract}


\maketitle

\vspace{-.07in}
\section{Introduction}
Online content platforms, such as YouTube, heavily rely on  recommendation systems to optimize user engagement on their platforms. 
For instance, 70\% of the content watched on YouTube is recommended by its algorithm \cite{rodriguez2018ytrecommendationsQZ}.
These recommendation systems provide personalized content recommendations by tracking and profiling user activity. 
For instance, YouTube tracks and profiles activities of its users on YouTube as well as off of YouTube to this end \cite{URL_YOUTUBE_SUPPORT}.
This tracking and profiling enables these platforms to predict relevant content that a user is likely to be interested in. 
On one hand, this tracking and profiling enables desirable utility to users by providing relevant content recommendations.
On the other hand, this tracking and profiling poses a privacy issue because the platform might infer potentially sensitive user interests.

Some platforms, including YouTube, allow users to remove a subset of the tracked activity (e.g., remove a specific video from YouTube watch history) or even disable the use of certain profiled user interests (e.g., gambling) to influence the recommendations. 
However, these controls do not necessarily stop the platform from
tracking and profiling user activities in the first place.
Thus, they may not provide much, if any, privacy benefit to users.
Moreover, the exercising of these controls would hurt the quality of personalized recommendations.
For example, if users employ these controls to curtail tracking or profiling then they will likely not receive personalized recommendations they are actually interested in.

The research community is increasingly interested in developing privacy-enhancing obfuscation approaches that do not rely on cooperation from online content platforms \cite{howe2017engineering,nissenbaum2009trackmenot,degeling2018tracking,zhang2021harpo}. 
%
%
%
%
At a high level, these privacy-enhancing approaches work by adding fake activity to real user activity to lessen the ability of the recommendation system to infer sensitive information.
However, the addition of fake activity for the sake of obfuscation also ends up impacting the utility users might derive from the recommendation system in terms of relevance of personalized recommendations. 
Prior obfuscation approaches attempt to navigate the trade-off between privacy and utility, for example \cite{zhang2021harpo}, by carefully adding fake activity so as to obfuscate ``private'' interests but allow ``non-private'' interests.

In this work, we are interested in designing a privacy-enhancing \textit{and} utility-preserving obfuscation approach for recommendation systems. 
In contrast to prior approaches that are typically limited to only obfuscating inputs to the recommendation system, our key idea is to design an obfuscation approach that can obfuscate inputs to preserve user privacy but at the same time remove ``noise'' from outputs to preserve the utility of recommendations. 
Since an adversarial recommendation system might also attempt to remove ``noise'', it is crucial that the denoiser can only be used by the user and not by the recommendation system. 
To this end, our insight is that the denoiser uses a ``secret'' input (specifically, a user's actual browsing history), which is only available to the user and not the recommendation system.  
The recommendation system instead only has access to the obfuscated browsing history of the user. 
Therefore, by leveraging the knowledge of a user's actual browsing history, the denoiser allows the user to preserve the recommendations related to the users' actual interests while discarding the unrelated recommendations caused by obfuscation.


We design and implement \ToolX, an obfuscation approach for YouTube's recommendation system that not only obfuscates a user's video watch history to protect privacy but then also denoises the video recommendations by YouTube to preserve their utility. 
\ToolX uses an \Obfuscator to inject obfuscation videos into a user's video watch history and a \Denoiser to remove recommended videos that are unrelated to the user's actual interests. 

The \Obfuscator is a RL model trained to insert YouTube videos in a users' watch history that will maximize the distortion in their interests being inferred by YouTube. 
We address three key issues in designing \ToolX's \Obfuscator, which is a non-trivial adaptation of Harpo \cite{zhang2021harpo} to YouTube. 
First, we build a surrogate of YouTube's recommendation system to efficiently train the RL model in a virtual environment. 
Second, we design the surrogate model to predict the distribution of hundreds of different classes of YouTube recommendation videos (we use the 154 affinity segments used by Google \cite{URL_GOOGLEADS} as our video classes) rather than the sheer number (order of hundreds of millions) of individual YouTube videos. Lastly, the \Obfuscator selects obfuscation videos based on embedding similarity, which is scalable to millions of obfuscation videos.

The \Denoiser is a ML model that is trained to reproduce the original recommendations that would have been received in the absence of the \Obfuscator. 
We address two key issues in designing \ToolX's \Denoiser.
First, \Denoiser makes use of a ``secret'' input (i.e., a user's actual watch history) as well as information that is also available to the adversarial recommendation system (i.e., obfuscated watch history and corresponding ``noisy" recommendations).
As we show later, this design ensures that only \ToolX is able to remove ``noise' while the adversary is unable to de-obfusacte without prohibitive collateral damage.  
Second, we define new divergence-based metrics to measure privacy and utility in training \Obfuscator and \Denoiser.

We deploy and evaluate \ToolX's effectiveness on YouTube using 10,000 sock puppet based personas, 10,000 Reddit user personas, and 936 real-world YouTube users \cite{casas2022exposure}. 
Our evaluation shows that \ToolX's \Obfuscator is able to degrade the quality of YouTube's recommendations by up to 87.23\% (privacy) and its \Denoiser is able to recover up to 90.40\% of the actual recommendations (utility).
We show that \ToolX outperforms the state-of-the-art by a factor of 2$\times$ in terms of improving utility for the same level of privacy. 
Crucially, we also demonstrate that \ToolX is stealthy and robust to de-obfuscation by an adversarial system. 
Our evaluation shows that the adversary incurs a prohibitively large number of false positives (order of tens/hundreds of millions) in attempting to undermine stealthiness and achieving  de-obfuscating.







\vspace{-.05in}
\section{Preliminaries}
\vspace{-.0in}
\subsection{Problem Statement}
\vspace{-.0in}
Recommendation systems track users' browsing activity to provide personalized recommendations.  
YouTube, for example, tracks users' browsing activity on YouTube (e.g., videos watched, channel subscriptions) as well as off of YouTube (e.g., activity on other Google services such as Google Search and Google Analytics, or web pages opened in Chrome browser) to personalize homepage and up-next video recommendations \cite{URL_YOUTUBE_SUPPORT}.
%
%
Users can selectively remove certain videos from their YouTube watch history or clear their browsing activity altogether to influence personalized video recommendations.
However, doing so does not necessarily mean that their browsing activity is not tracked in the first place, and thus there is no material privacy benefit to users.
It will also hurt the quality of personalized recommendations because users will likely not receive recommendations for videos they are interested in. 
In summary, users are unable to exert meaningful control over recommendation systems to protect their privacy while preserving the utility of personalized recommendations. 


Prior work has proposed obfuscation approaches to protect user privacy in personalized recommendation systems without relying on cooperation from online content platforms. 
Existing approaches obfuscate a user's browsing history by injecting fake activity (e.g., webpage visits) to manipulate a user's interest segments and targeted ads in online behavioral advertising \cite{zhang2021harpo,xing2013take}. 
These obfuscation approaches are designed for recommendation systems (e.g., online behavioral advertising) where users are not necessarily interested in consuming the output of the recommendation system, rather users are mainly interested in subverting it. 
While these approaches aim to protect user privacy (e.g., inferred interest segments), they do not consider the utility of recommendations (e.g., whether targeted ads are of interest to the user). 
In contrast, in recommendation systems such as YouTube, these obfuscation tools would render the utility of YouTube's video recommendations useless to the user.

\textit{Can we design privacy-enhancing obfuscation approaches that can enhance privacy of users and at the same time preserve utility for users in recommendation systems?}
With this goal in mind, we propose to build a denoiser to remove the ``noisy" videos injected as part of obfuscation. 
It is crucial that the denoiser can only be used by the user and not by the recommendation system. 
To this end, our insight is that the denoiser uses a ``secret'' (specifically, the user's actual browsing history), which is only available to the user and not the recommendation system.  
Therefore, by leveraging the knowledge of a user's actual browsing history, the denoiser may preserve the recommendations related to the users' actual interests while discarding the unrelated recommendations caused by obfuscation. 
Figure \ref{fig:problemwithwithout} illustrates this idea that we next operationalize in \ToolX.

\begin{figure}[t]
\begin{subfigure}{.5\textwidth}
\centering
    \includegraphics[width=.83\linewidth]{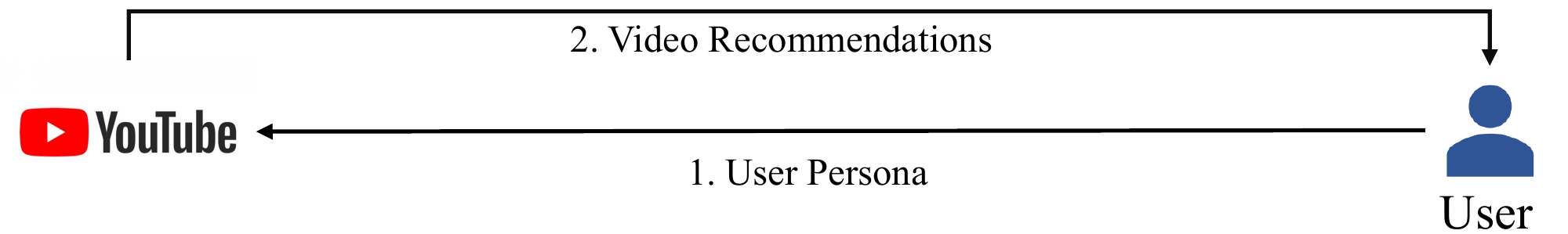}
    \caption{Without obfuscation-denoising system.}
    \label{fig:l2_budget}
\end{subfigure}
\begin{subfigure}{.5\textwidth}
\centering
    \includegraphics[width=0.83\linewidth]{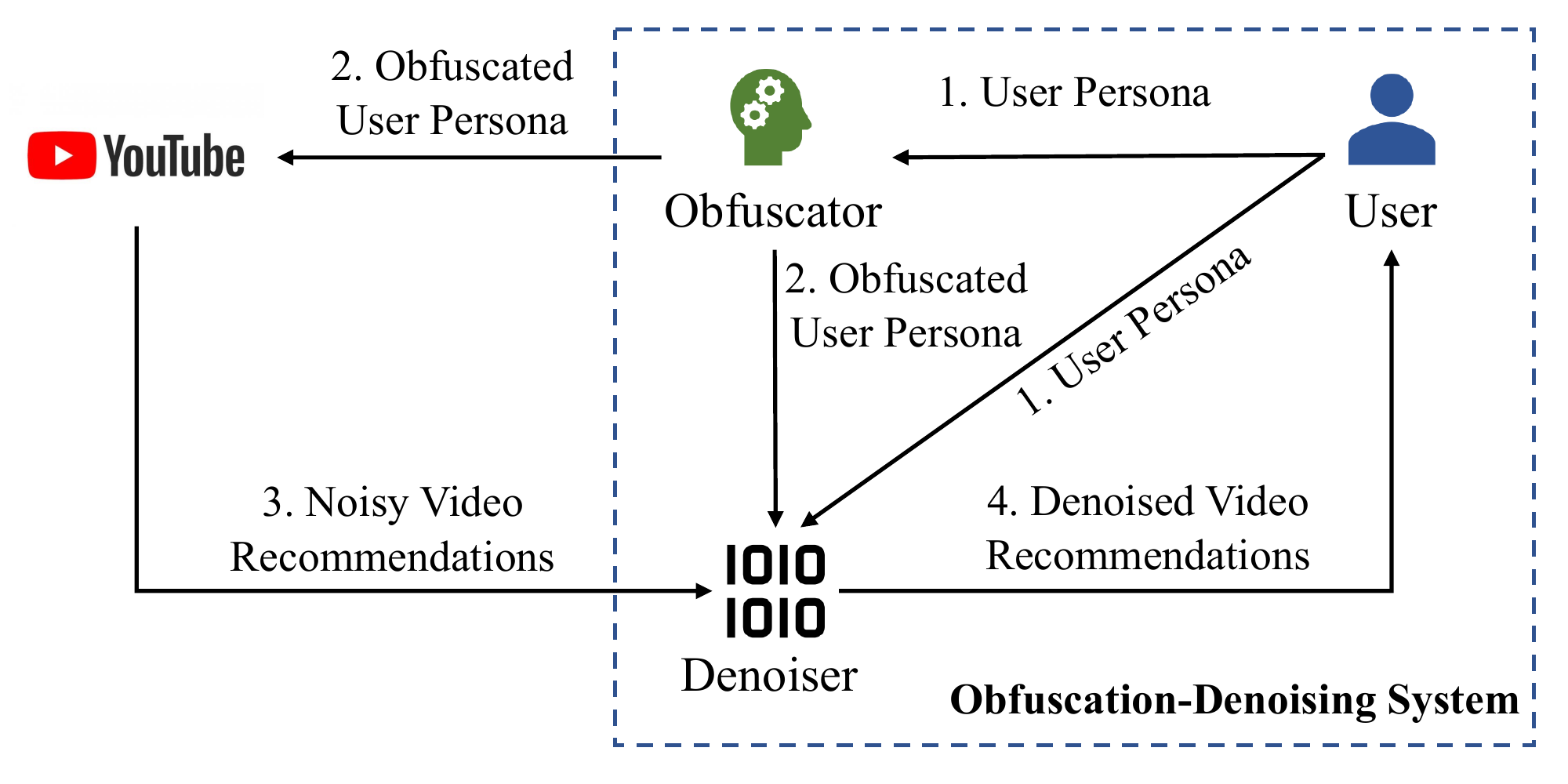}
    \caption{With obfuscation-denoising system.}
    \label{fig:problem}  
\end{subfigure}
\vspace{-.15in}
\caption{Problem Overview.}
\vspace{-.22in}
\label{fig:problemwithwithout}
\end{figure}

\vspace{-.1in}
\subsection{Threat Model}
\label{sec:threatmodel}
\vspace{-.05in}
\paratitle{User.}
The user's goal is to routinely browse YouTube videos and get high-quality recommendation videos fitting their interests, while misleading the YouTube recommendation system such that it can not accurately infer the user's interests. 
To achieve this goal, users install a local obfuscation-denoising system, which consists of an \Obfuscator  and a \Denoiser. The \Obfuscator will obfuscate their video watching history by injecting fake video watches into the user's real video watches, and the \Denoiser will automatically remove ``noisy" recommended videos from YouTube (i.e. caused by obfuscation) that do not fit user's interests. The obfuscation-denoising system is designed to satisfy the following properties:
\begin{itemize}
    \item it is \textbf{privacy-preserving} in that the user's interests are protected from being inferred by YouTube.
    \item it is \textbf{utility-preserving} in that the user can receive high-quality videos fitting their interests.
    \item t has \textbf{low overhead} in that the amount of obfuscation videos inject will not affect the user experience.
    \item it is \textbf{stealthy} in that it is impractical for YouTube to detect the usage of obfuscation-denoising system.
    \item it is \textbf{robust to deobfuscation} in that it is impossible for YouTube to distinguish fake video watches from real video watches.%
    \item it can be \textbf{personalized} in that it can treat video classes differently based on user preferences.
\end{itemize}
\vspace{-.02in}
\paratitle{Recommendation system.}
The goal of the recommendation system is to track user activity for personalized recommendations to maximize user engagement (e.g., click rate and watch time). We assume that the recommendation system has full access to the user's video watching history (including both fake and real video watches though it does not know which is which) and it recommends videos based on the user's video watching history, which is true for YouTube \cite{URL_YOUTUBE_BLOG} (unless the user deletes their watching history). 
We further assume that the recommendation system does not have access to the user's off-platform browsing history (e.g., the user is not simultaneously signed-in to YouTube and other services by YouTube's parent company Google, the user employs Google account controls to prevent off-YouTube information linking (if the user is signed-in to YouTube and other services by YouTube's parent company Google) \cite{WebActivityControls}, or the user uses a browser such as Safari \cite{safari-cookies-blocked} or Firefox \cite{firefox-cookies-blocked} -- or privacy-enhancing browser extension \cite{uBlock-Origin} -- that prevents cross-site tracking). 
We also assume that the recommendation system has substantial computation resources to train a machine learning model for its recommendations. This assumption also holds for YouTube \cite{covington2016deep}. Moreover, we assume that the recommendation system has access to \ToolX once it is public, such that it can use it to analyze the obfuscation approach and possibly train adversarial detectors to detect and filter the usage of \ToolX. More specifically, we assume that the recommendation system has a two-step detection workflow. In the first step, the adversary will train a classifier to detect whether or not a user uses \ToolX. Then, in the second step, if \ToolX usage is detected, the adversary further attempts to achieve deobfuscation by filtering out obfuscation videos and keeping the remaining videos.



\vspace{-.05in}
\section{Proposed Approach}
\label{sec:approach}
\vspace{-.0in}
In this section, we present the proposed utility-preserving obfuscation approach \ToolX.
\vspace{-.05in}
\subsection{Overview}
\label{subsec:sys_overview}
\vspace{-.05in}
As already discussed, at a high-level \ToolX consists of an \Obfuscator designed for enhancing user privacy and a \Denoiser designed for preserving user utility, as demonstrated in Figures \ref{fig:problemwithwithout} and \ref{fig:overalldetails} (in more detail). 
The \ToolX \Obfuscator is a non-trivial adaptation of Harpo's obfuscator \cite{zhang2021harpo} in the context of YouTube's recommendation system.
The \Obfuscator injects fake video playing records into a user's video playing history at random times. We refer to videos played by the user as \textit{user videos} and to videos played by the \Obfuscator as \textit{obfuscation videos}.
Note that without any \textit{obfuscation videos} in the user's video playing history (which is denoted by $V^u$ in this case),  YouTube will recommend a set of videos desired by the user. We refer to this set of videos as ``clean'' YouTube videos.
However, with \textit{obfuscation videos} in the user's video playing history (which is denoted by $V^o$ in this case), YouTube will recommend a set of videos which include videos undesired by the user. We refer to this set of videos as ``noisy'' YouTube videos.
The \Denoiser is designed to predict the class distribution
of ``clean'' YouTube videos from the class distribution of ``noisy'' YouTube videos, such that \ToolX can repopulate a new set of videos with the same class distribution as the ``clean'' YouTube videos.  We refer to the repopulated videos as \ToolX videos. Note that each video class represents a video topic, and we use the 154 affinity segments used by Google \cite{URL_GOOGLEADS} as our video classes.

In more detail, see Figure \ref{fig:overalldetails},
\ToolX starts by generating video embeddings of past played videos via an embedding model. It then uses an \Obfuscator model to select obfuscation videos based on the generated video embeddings. Note that we follow a similar methodology with that in \cite{zhang2021harpo} to formulate the process of inserting \textit{obfuscation videos} as a Markov Decision Process (MDP), and use reinforcement learning (RL) to train the \Obfuscator model to maximize the divergence between the class distribution of ``noisy'' YouTube videos (denoted by $C^o$) and the class distribution of ``clean'' YouTube videos (denoted by $C^u$).
After receiving the ``noisy'' YouTube videos, the \Denoiser outputs an estimate of the class distribution of ``clean'' YouTube videos (denoted by $\hat C^u$), by taking as inputs $V^u$, $V^o$, and $C^o$. Finally, \ToolX will use a repopulation model to generate the set of \ToolX videos with class distribution $\hat C^u$.

\begin{figure}[t]
\begin{subfigure}{.5\textwidth}
    \includegraphics[width=.85\linewidth]{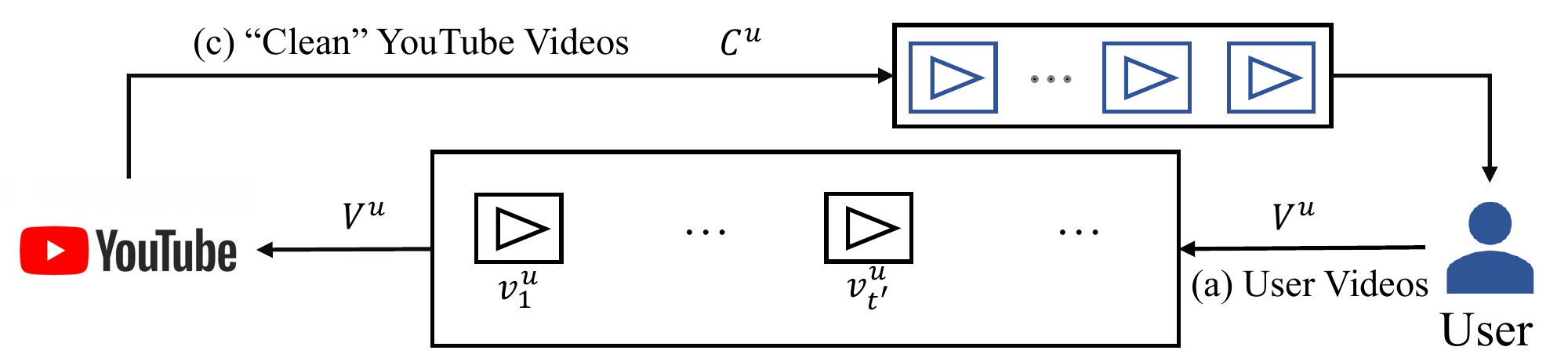}
    \caption{Without obfuscation-denoising system.}
    \label{fig:overall_org}
\end{subfigure}
\begin{subfigure}{.5\textwidth}
    \includegraphics[width=.85\linewidth]{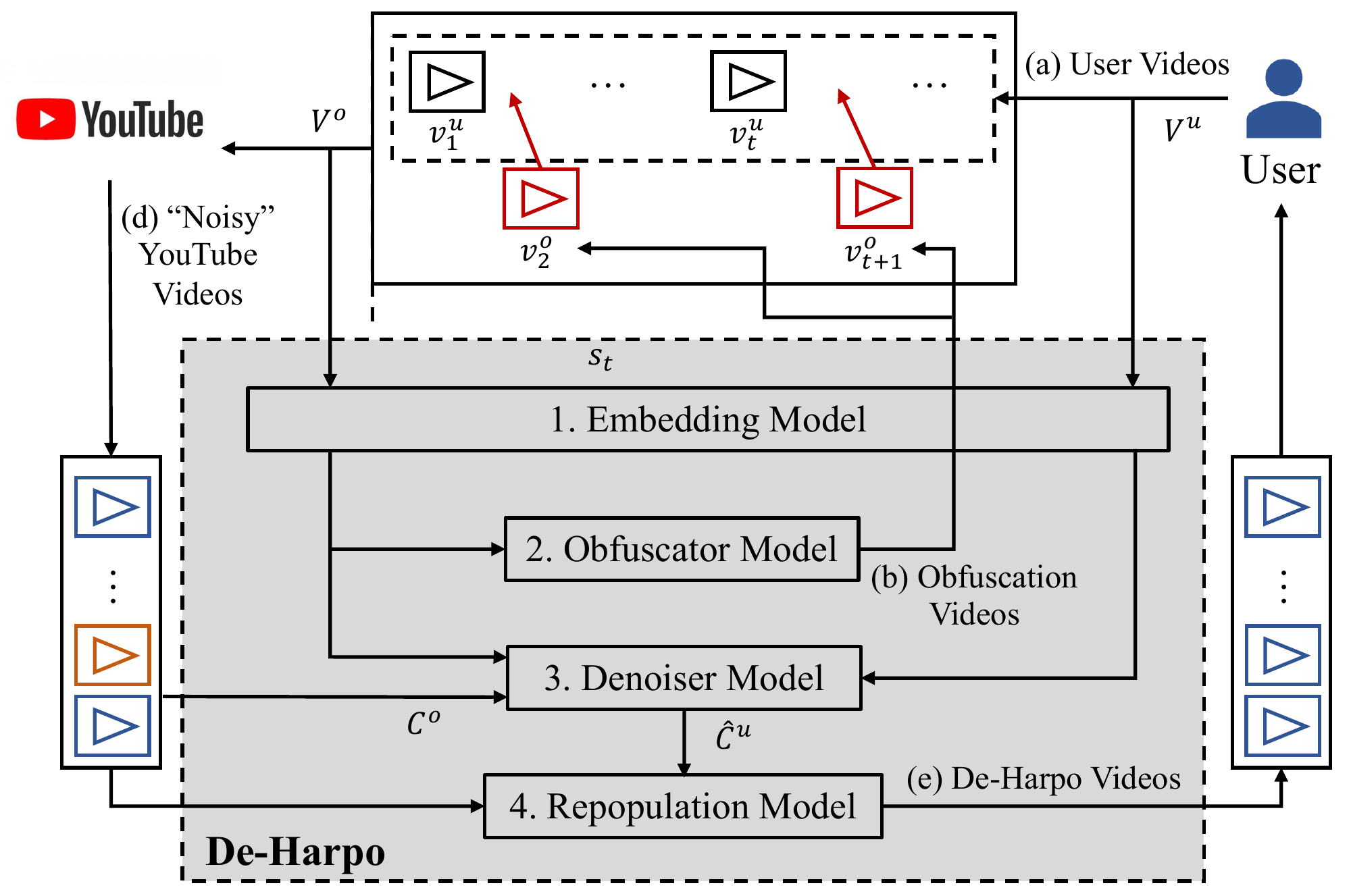}
    \caption{With obfuscation-denoising system.}
    \label{fig:overall}  
\end{subfigure}
\vspace{-.15in}
\caption{Overview of \ToolX. Note that $V^u$ denotes the non-obfuscated user persona, $V^o$ denotes the obfuscated user persona generated by the \Obfuscator, $C^u$ is the recommended video class distribution based on $V^u$, $C^o$ is the recommended video class distribution based on $V^o$, $\hat C^u$ is the \Denoiser's estimate of $C^u$, and $v_i^u$ and $v_i^o$ represent user video and obfuscation video respectively.}
\vspace{-.2in}
\label{fig:overalldetails}
\end{figure}

\vspace{-.10in}
\subsection{System Preliminaries}
\label{subsec:preliminaries}
\vspace{-.05in}
\paratitle{User persona.} We define a user persona as a sequence of YouTube videos. Formally, we denote the non-obfuscated user persona as $V^u=[v_1^u,...,v_n^u]$, where $v_i^u$ represents the $i$th video played by the user, and $n$ is the total number of videos played by the user. We denote the obfuscated user persona as $V^o=[v_1^j,...,v_{n'}^j]$, where $j\in\{u,o\}$, $v_i^u$ and $v_i^o$ represent that the $i$th video is played by the user and \Obfuscator respectively, and $n'$ is the total number of videos played by the user and \Obfuscator combined. 

\vspace{-.02in}
\paratitle{Recommended video class distribution.}
\label{subsubsec:rec_video_dist}
We define the recommended video class distribution of a non-obfuscated user persona $V^u$ (i.e. the class distribution of ``clean'' YouTube videos) as $C^u=[c_1^u,...,c_K^u]$, where $\sum_{k=1}^{k=K}c_k^u=1$, $c_k^u$ is the percentile of videos from the $k$th class among recommended videos for $V^u$, and $K$ is the total number of classes.  Similarly, we define the recommended video class distribution of an obfuscated user persona $V^o$ (i.e. the class distribution of ``noisy'' YouTube videos) as $C^o=[c_1^o,...,c_K^o]$, where $\sum_{k=1}^{k=K}c_k^o=1$ and $c_k^o$ is the percentile of videos from the $k$th class among the recommended videos for $V^o$. 
We use the recommended video class distribution as a representation of the user interest profile built by YouTube instead of directly using the recommended videos. This design choice is made to (i) mitigate the impact of non-determinism in
YouTube’s recommendations and (ii) alleviate the difficulty of making video-level recommendations given an incomplete set of available videos while still making reasonably fine-grained recommendations (among 154 different classes).

\paratitle{Privacy metric.} At a high level, we want to distort the user interest profile built by YouTube for user personas to enhance user privacy. Motivated by the use of the recommended video class distribution as a representation of YouTube's user interest profile, we first define the following privacy metric:
\begin{equation}
\label{eq:priavcy}
    P=E[D_{KL}(C^o||C^u)]=E[\sum_{k=1}^{k=K}c_k^o\log\frac{c_k^o}{c_k^u}],
    \vspace{-.05in}
\end{equation}
which measures the expected KL divergence between the two probability distributions ($C^o$ and $C^u$)\footnote{Note that if $c_k^i=0$ ($i\in\{u,o\}$), we assign a small value to it to avoid getting $\infty$ in KL divergence calculation.}. 

It is worth noting that we use KL divergence since it is a well-established measure of the discrepancy between two distributions, and, together with the closely related mutual information measure they have been used as on-average privacy metrics in myriad of applications including recommendation systems \cite{cuff2016differential,clark2020optimizing,zhang2022privacy,elkordy2022much,parra2014measuring,parra2017myadchoices}. We do not use stricter privacy metrics which provide worst-case privacy guarantees (e.g. differential privacy (DP) \cite{dwork2014algorithmic}), since in the context of our application one would need to inject an enormous number of obfuscation videos to satisfy such guarantees (see Section \ref{subsec:goal} and Appendix \ref{appendix:dp} for a detailed, formal discussion on DP in our context).

During real-world experimentation on YouTube, we observe that the recommended video class distribution of the same persona may differ a bit due to an inherent randomness of the system. 
Since we are interested to measure the divergence thanks to obfuscation only, we define $D^{Min}$ as the expected KL divergence between a random sample of $C^u$ and its mean $\bar C^u$ (i.e., $D^{Min}=E[D_{KL}(\bar C^u, C^u)]$), and subtract from $P$ the divergence caused by randomness, that is, we work with $P - D^{Min}$. 
Furthermore, since $P$ is unbounded, we normalize the privacy metric as follows. Denote the user persona set as $\mathcal{V}$, which consists of all user personas. Let $V^u$ and $V^{u^{'}}$ be two user personas uniformly and randomly sampled from $\mathcal{V}$, and let their associated recommended video class distributions be $C^u$ and $C^{u^{'}}$ respectively. Then, we define the normalized privacy metric $P^{Norm}$ by:
\begin{equation}
\label{eq:priavcy_norm}
    P^{Norm}=\frac{P - D^{Min}}{D^{Max} - D^{Min}},
    \vspace{-.05in}
\end{equation}
where $D^{Max}=E[D_{KL}(C^u,C^{u'})]$ is the expectation of the KL divergence between $C^u$ and $C^{u'}$ and thus corresponds to the average ``distance" between two video class distributions of two randomly selected users. Hence, $P^{Norm}$ measures the fraction of the maximum possible divergence that obfuscation achieves, on average. Note that for both $P$ and $P^{Norm}$, the higher their value is, the more effective the \Obfuscator is in enhancing user privacy (see Figure \ref{fig:metric}). 

\paratitle{Utility metric.} In our threat model, the user sends the obfuscated persona to YouTube and then receives a ``noisy'' recommended video list with class distribution $C^o$. However, the user desires the ``clean'' recommended video list with class distribution $C^u$. Our \Denoiser is designed to predict $C^u$ from $C^o$, such that \ToolX can repopulate the ``clean'' recommended video list from $C^u$. With the above in mind, we define our utility loss metric as follows:
\begin{equation}
\label{eq:utility_loss}
U_{Loss}=E[D_{KL}(\hat C^u||C^u)]=E[\sum_{k=1}^{k=K}\hat c_k^u\log\frac{\hat c_k^u}{c_k^u}],
\vspace{-.05in}
\end{equation}
where $\hat C^u$ is the output of the \Denoiser, representing its estimation of $C^u$. Smaller $U_{Loss}$ means smaller divergence between the non-obfuscated recommended video class distribution $C^u$ and the \Denoiser's estimate of such distribution $\hat C^u$ and thus a better estimate. The theoretical minimum that this value can take is 0, representing two identical distributions i.e. the noise is perfectly removed.
Note that without applying the \Denoiser, the utility loss equals the value of privacy $P$ (since $\hat C^u= C^o$). The \Denoiser can reduce the utility loss caused by the \Obfuscator by $P-U_{Loss}$ which represents the \Denoiser utility gain. Similarly to above, because $P$ is unbounded and YouTube's randomness causes, on average, a divergence of $D^{Min}$, we define the normalized utility gain metric as follows:
\begin{equation}
\label{eq:utility_gain}
\vspace{-.05in}
    U_{Gain}^{Norm}=\frac{P-U_{Loss}}{P-D^{Min}},
\end{equation}
which represents the fraction of obfuscation noise reduced by the \Denoiser, on average. Higher $U_{Gain}^{Norm}$ implies that the \Denoiser can reduce the utility loss caused by the \Obfuscator more effectively and a value of 100\% indicates a complete removal of noise (see Figure \ref{fig:metric}).

\begin{figure}[h]
    \centering
    \includegraphics[width=.80\linewidth]{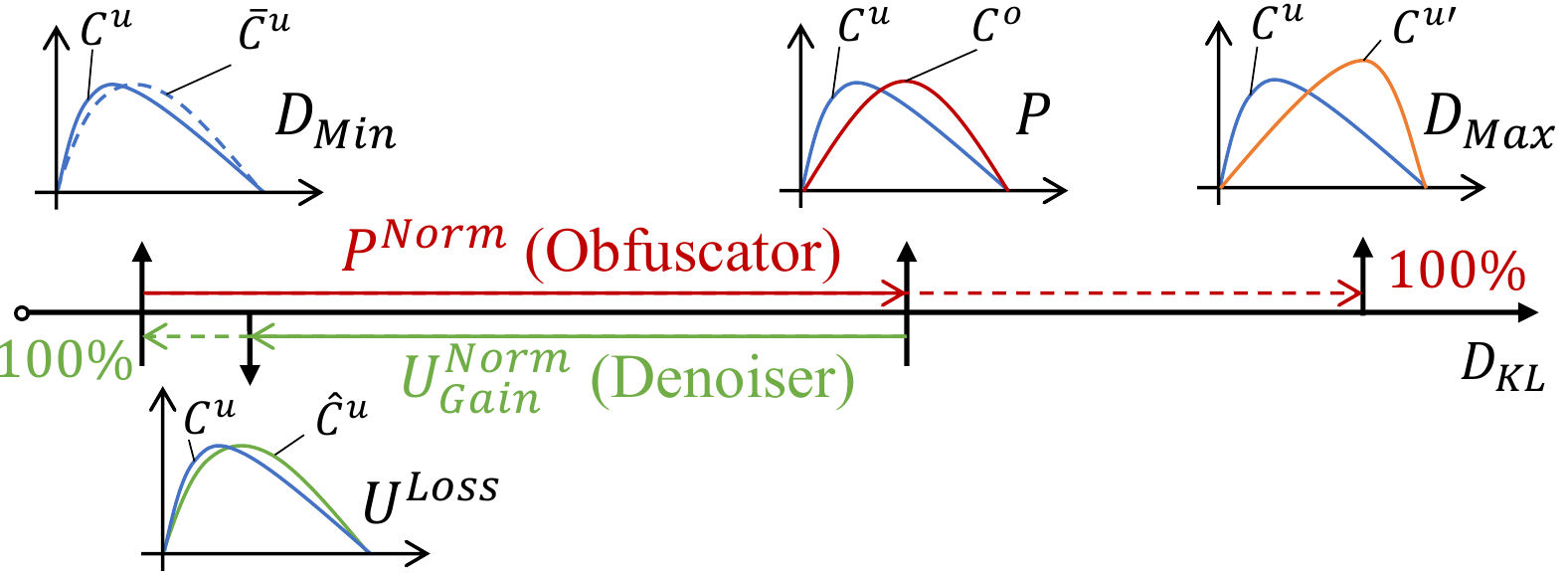}
    \vspace{-.13in}
    \caption{Privacy and utility metrics.}
    \vspace{-.20in}
    \label{fig:metric}
\end{figure}

\vspace{-.05in}
\subsection{Performance Goals and Guarantees}
\label{subsec:goal}
\vspace{-.02in}

\paratitle{Performance goals.}
As discussed already, our goal is to obfuscate the actual user profile, that is, the inferred user’s interests by YouTube from the user’s video watch history. (We do not consider other sub-channels via which YouTube may infer user interests, see threat model details in Section \ref{sec:threatmodel}.) 
In view of obfuscation, YouTube’s goal is to reconstruct the actual user profile (what YouTube would have inferred by the user’s video watch history in the absence of obfuscation) as accurately as possible from the obfuscated user profile (what YouTube infers by the user’s video watch history in the presence of obfuscation).
Since YouTube’s user profiles are not public, we infer them from YouTube's recommended videos to the user, and, more specifically, from the the recommended video class distribution (where we use the 154 affinity segments used by Google as our video classes).

Motivated by the above, our privacy metric maximizes the distance (normalized KL divergence) between the recommended video class distribution before and after obfuscation. If the distance between the recommended video class distribution before and after obfuscation is almost the same with the distance between the recommended video class distribution before obfuscation and the recommended video class distribution of another random user, then YouTube’s recommendations for the user under study are essentially random implying that YouTube is not able to learn the user’s actual interests from the obfuscated user’s video watch history. Tellingly, in Section \ref{subsec:tradeoff} we do show that with merely 70\% of videos in a user persona being obfuscation videos, the distance between the recommended video class distribution before and after obfuscation is already 93\% of the distance between random distributions.

\vspace{-.02in}
\paratitle{Performance guarantees.}
A discussion about performance guarantees is in order.
First, can \ToolX effectively de-noise the noisy recommendations such that their utility is high, despite that recommendations are as if they were random? Section \ref{subsec:utility} answers affirmatively. Related to this, if De-Harpo can de-noise recommendations, can't YouTube de-noise them as well? Sections \ref{subsec:stealthy}, \ref{subsec:robust} show that it cannot in practice, and Section \ref{subsec:secret} offers a formal explanation why it can't. Note that even though YouTube unavoidably learns the interests of a user corresponding to the user videos that the user actually watches, it also learns interests corresponding to the obfuscation videos, the relative importance of each interest is altered, and YouTube has no practical way of telling which interest is real and which is not.~\footnote{If a user wishes YouTube to not learn about the user's real interests at all, the user should not use YouTube: Even though YouTube in theory offers a method to remove a video from the watch history, (i) even if the video is deleted the corresponding interest categories are not \cite{Mozilla-Report-YouTube} and (ii) there is no "unlearning" at the ML level and hence the recommendation algorithm will still recommend videos based on the total watch history.}

Second, both our privacy and utility metrics are based on expectations, see Eq. (1)-(4). Hence, \ToolX guarantees performance goals ``on-average". But what about ``worst-case" privacy guarantees? In our context this would require that no matter how unique the original video watch history of a specific user may be, YouTube should not be able to learn any unique interests of this user regardless of how unsuccessful it may be across all users on average. There is a large line of prior work on both ``on-average" \cite{clark2020optimizing,zhang2022privacy,elkordy2022much,parra2014measuring,parra2017myadchoices} and ``worst-case" \cite{dwork2006differential,dwork2014algorithmic,cuff2016differential,zhang2022privacy,elkordy2022much} privacy guarantees. It is intuitive that strict definitions of privacy like differential privacy (DP) \cite{dwork2014algorithmic}, which guarantee privacy in the worst-case, cannot be satisfied for recommendations systems actively used by users. For a matter of completeness, we provide a formal proof about why differential privacy can not be achieved in Appendix \ref{appendix:dp}. A summary of the argument follows:
Assume that there is one video $V$ in user persona $P_1$ (i.e. video watch history) which is not in user persona $P_2$, and the obfuscator $O$ (the randomized function in the DP definition) can not remove it from $P_1$. Let $P$ be a user persona without video $V$ that we observe. Then, the probability of $O(P_1)$ being $P$ is zero while the probability of $O(P_2)$ being $P$ is non-zero. Thus, per the DP definition, the $\epsilon$ for this worst-case scenario will be infinite and DP is violated. 

\vspace{-.08in}
\subsection{System Model}
\label{subsec:sys_model}
\vspace{-.02in}
\paratitle{Obfuscator.}
The obfuscation video selection process of \Obfuscator is formulated as a Markov Decision Process (MDP). 
At the beginning of each time step, a video will be played. If the played video is an obfuscation video injected by the \Obfuscator, we refer to this time step as an obfuscation step.  
Let $\alpha\in [0,1)$ be the obfuscation budget which we use as a system parameter to control the percentile of obfuscation videos.
At each time step, with probability $\alpha$ an obfuscation video will be injected by \Obfuscator into the user persona.
Let $s_t$ denote the state of the MDP at obfuscation step $t$, defined as all the played videos until now, and $a_t$ denote the action of the MDP at obfuscation step $t$, which represents the obfuscation video sampled based on the MDP policy.
The MDP policy is a probability distribution which outputs the probability of selecting obfuscation video $i$ given state $s_t$, and we denote this probability by $\pi(a_t=i|s_t)$. 
We associate a reward $r_t$ for the action $a_t$ at obfuscation step $t$.
We set the reward $r_t$ to $P_t-P_{t-1}$, where $P_t$ is the privacy metric value at obfuscation step $t$.
The goal of solving this MDP is to find the optimal policy, such that the accumulative rewards $\sum_{t=1}^{t=T}r_t$ can be maximized. Note that $T$ denotes the total number of obfuscation steps. (We consider a finite-horizon MDP)
Appendix \ref{appendix:mdp} discusses the MDP in more detail.

\vspace{-.02in}
\paratitle{Denoiser.} At a high level, we model the \Denoiser as a mapping from the recommended video class distribution of the obfuscated user persona $C^o\in\mathbb{R}^K$ to the recommended video class distribution of the non-obfuscated user persona $C^u\in\mathbb{R}^K$ ($K$ is the total number of video categories). 

Estimating directly $C^u$ from $C^o$ can be challenging. In the extreme case, where the mutual information between $C^u$ and $C^o$ is zero \cite{mackay2003information}, it is impossible for the \Denoiser to estimate $C^u$ from $C^o$. To estimate $C^u$, the \Denoiser may leverage side information indicating how the obfuscation videos are injected into the user personas, as in this case it may be able to undo the effect of obfuscation videos in the recommendations list. 
In our application, such side information is explicitly available to users ($V^u$ portion of $V^o$), since the \Obfuscator is installed locally and users know exactly how the obfuscation videos are injected into user personas. Therefore, our \Denoiser is modeled to be a functional mapping from $(V^u, V^o, C^o)$ to $C^u$.

\subsection{The ``Secret" of the Denoiser} 
\label{subsec:secret}
We use the information theory concept of mutual information (MI) to explain why the \Denoiser works.
Recall that the recommendation system cannot distinguish user from obfuscation videos thus does not know the user's video playing history $V^u$.
In our system, both $V^u$ and $V^o$ are modelled as random vectors, and  $V^o$ is generated from $V^u$ by the \Obfuscator, which is a random function. Additionally, both $C^u$ and $C^o$ are random vectors, which are generated from $V^u$ and $V^o$ respectively by the YouTube recommendation system. By applying the chain rule of MI, we can derive the following equation:
\vspace{-.05in}
\begin{align}
\vspace{-.1in}
\label{eq:mi1}
    I(C^o,V^o,V^u;C^u) = I(C^o,V^o;C^u) + I(V^u;C^u|C^o,V^o),
\end{align}
where $I(C^o,V^o,V^u;C^u)$ is the MI between $(C^o,V^o,V^u)$ and $C^u$, $I(C^o,V^o;C^u)$ is the MI between $(C^o,V^o)$ and $C^u$, and $I(V^u;C^u|C^o,V^o)$ is the MI between $V^u$ and $C^u$ conditioning on $(C^o,V^o)$.

First, we show that the non-obfuscated user persona $V^u$ can be leveraged by the \Denoiser to better estimate $C^u$. 
Since $C^u$ is generated by YouTube recommendation system given $V^u$, $V^u$ is correlated with $C^u$, thus $I(V^u;C^u|C^o,V^o)>0$. Hence,
\vspace{-.05in}
\begin{align}
\vspace{-.1in}
\label{eq:mi_ineq_1}
I(\underbrace{C^o,V^o,V^u}_{\mbox{\scriptsize with secret}};C^u) > I(\underbrace{C^o,V^o}_{\mbox{\scriptsize without secret}};C^u).
\end{align}
Since the MI between $(V^u, V^o, C^o)$ and $C^u$ is larger than the MI between $(C^o,V^o)$ and $C^u$, $(C^o,V^o,V^u)$ can reveal more information about $C^u$ than $(C^o,V^o)$, leading to a more accurate estimate of $C^u$. As an aside, note that YouTube may attempt to de-obfuscate $V^u$ from $V^o$. We evaluate the robustness of the \Obfuscator against de-obfuscation in Section \ref{subsec:robust}.

Second, we show that including $C^o$ and $V^o$ may help to further enhance the effectiveness of the \Denoiser, compared with using $V^u$ only. Based on the chain rule of MI, we can rewrite Eq. (\ref{eq:mi1}) as follows:
\vspace{-.05in}
\begin{align}
\vspace{-.1in}
\label{eq:mi2}
    & I(V^u,V^o,C^o;C^u)\nonumber \\
    & = I(V^u;C^u) + I(C^o;C^u|V^u)+ I(V^o;C^u|C^o,V^u).
\end{align}
Consider the term $I(C^o;C^u|V^u)$. 
$C^o$ depends on $V^u$ and the obfuscation videos, and $C^u$ depends on $V^u$. 
Crucially, they both also depend on the (non deterministic) YouTube recommendation system. 
Hence, even when $V^u$ is given, there is non-zero MI between $C^o$ and $C^u$, that is, $I(C^o;C^u|V^u)>0$, leading to the following inequality:
\vspace{-.05in}
\begin{align}
\vspace{-.1in}
\label{eq:mi_ineq_2}
I(V^u,V^o,C^o;C^u) > I(V^u;C^u),
\end{align}
which means the MI between $(V^u, V^o, C^o)$ and $C^u$ is larger than the MI between $V^u$ and $C^u$ only. Intuitively, knowing the pair $V^o,C^o$ reveals information about how the YouTube recommendation system selects videos to recommend given a user video watching history. 
Therefore, the \Denoiser taking $C^o$ and $V^o$ as additional inputs can learn more information about $C^u$, as compared to the \Denoiser taking only $V^u$ as input. Our evaluation results in Section \ref{subsec:utility} empirically support the above analysis.

\begin{figure}[!t]
\centering
\begin{subfigure}{.21\textwidth}
\centering
    \includegraphics[width=.99\linewidth]{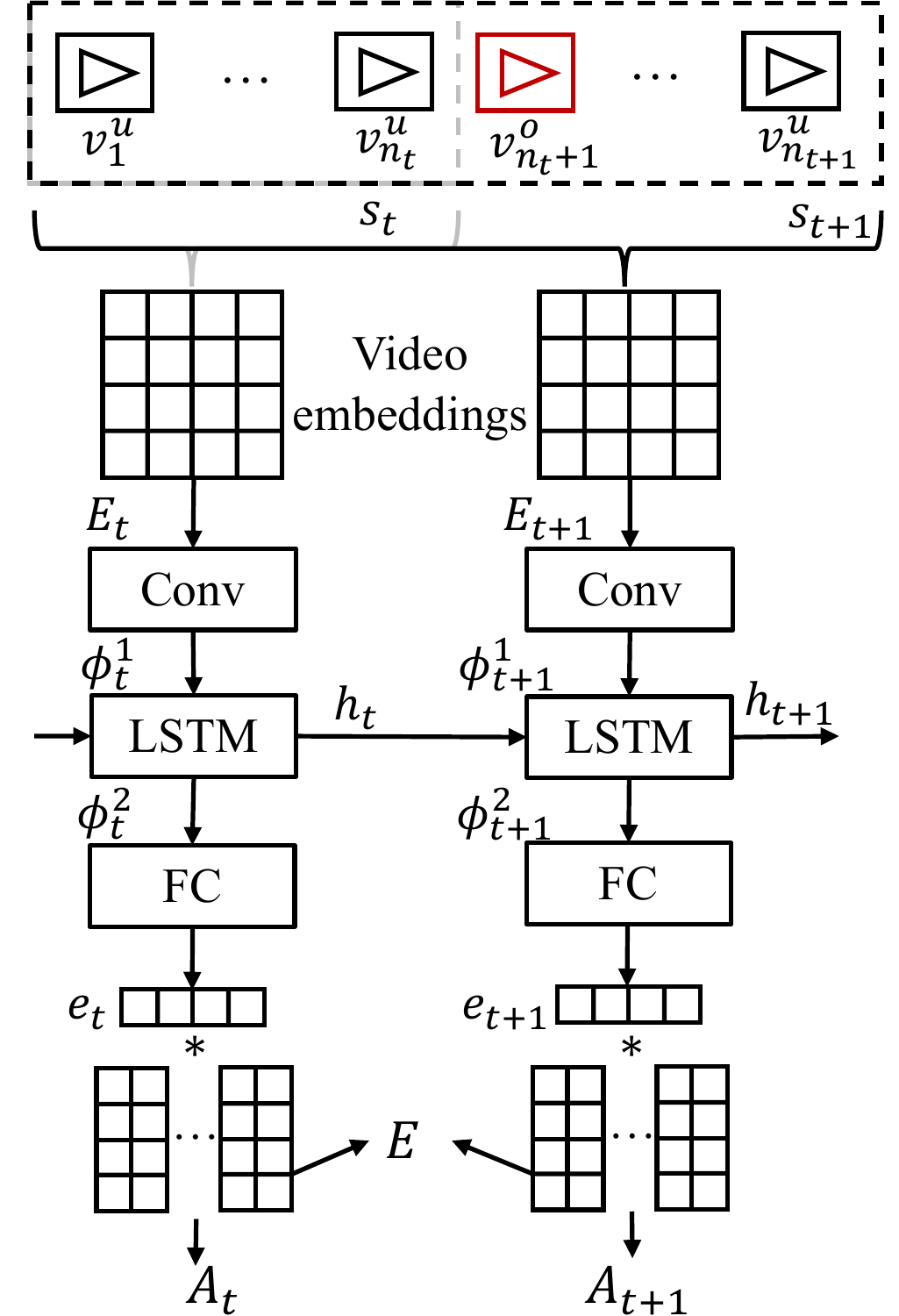}
    \caption{\Obfuscator}
    \label{fig:stru2}
\end{subfigure}
\begin{subfigure}{.21\textwidth}
\centering
    \includegraphics[width=.99\linewidth]{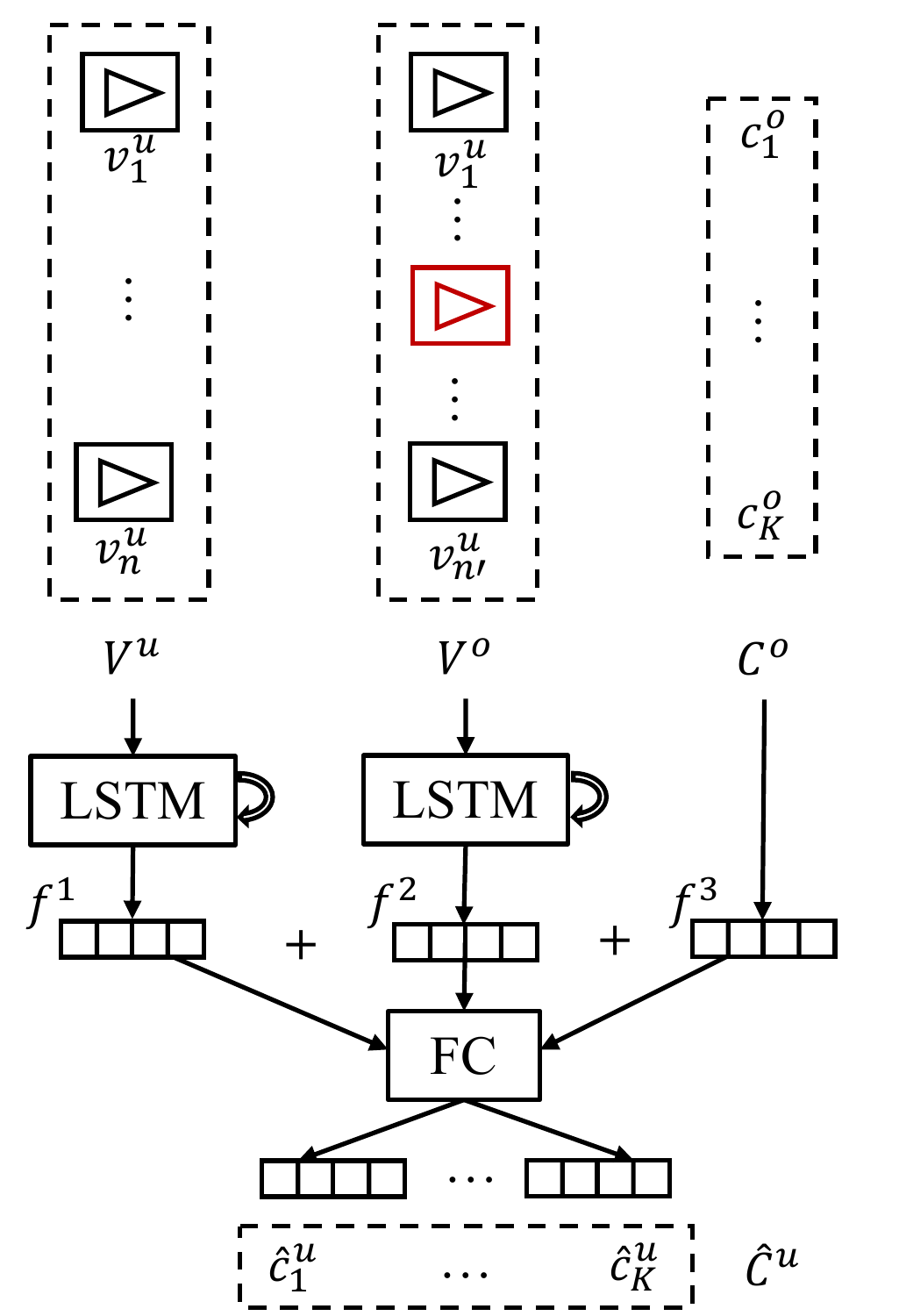}
    \caption{\Denoiser}
    \label{fig:stru3}
\end{subfigure}
\vspace{-.15in}
\caption{Details of system design. } 
\vspace{-.15in}
\label{fig:budget}
\end{figure}

\vspace{-.0in}
\section{System Design and Implementation}
In this section, we describe the detailed design of \ToolX and how we implement \ToolX as a browser extension. \ToolX consists of five modules: (1) a video embedding model that maps videos into embeddings; (2) a \Obfuscator model that selects obfuscation videos based on the video embeddings of played videos;  (3) a \Denoiser model that estimates the class distribution of ``clean" YouTube videos from the class distribution of ``noisy" YouTube videos; (4) a repopulation model that outputs \ToolX videos with the estimated class distribution of ``clean" YouTube videos; (5) a surrogate model used to train the \Obfuscator model offline efficiently (see Figure \ref{fig:overall} for the workflow of modules (1)-(4)).

\vspace{-.03in}
\subsection{Video Embedding}
\label{subsubsec:video_emb}
To make our system scalable to millions of YouTube videos without being restricted to a fixed set,
we represent each video by an embedding vector. A YouTube video typically consists of metadata (e.g. title, description, view count, rating, thumbnail, etc), a sequence of image frames (i.e. the video), and the transcript for the video. Since a video's transcript is a good representation of its content and it is more computationally and spatially efficient to process the transcript compared to processing the original video stream, we use video metadata and transcript to generate the video embedding, where the video embedding for video $v_i$ is denoted by $e_i\in\mathbb{R}^{404}$ (see Figure \ref{fig:stru1} in Appendix \ref{appendix:video_emb} for details) \footnote{Note that the YouTube recommendation system will use the image frames and some other private features to generate the video embedding (see \cite{covington2016deep}). We acknowledge that by including these features, our video embeddings may be closer to the actual embeddings used by YouTube. However, since our video embeddings can already yield a surrogate model (see Section \ref{subsubsec:surr_model} with reasonable performance and it is more computationally efficient, we choose the current design of our video embeddings.}. 

\vspace{-.03in}
\subsection{Obfuscator Model} 
\label{subsubsec:obfu_model}
As discussed before, we model the process of injecting obfuscation videos as an MDP.
Due to the prohibitively large state space of this MDP, we use RL, parameterized by a deep neural network, to learn the optimal policy for obfuscation video selection. 

The \Obfuscator takes as input the state at each 
obfuscation step, and outputs a video embedding. 
By measuring the cosine similarity between the output video embedding and each obfuscation video embedding, the \Obfuscator derives the probability distribution of the obfuscation video selection, where an obfuscation video with more similar embedding as the output video embedding is assigned a higher probability. Specifically, as shown in Figure \ref{fig:stru2}, the \Obfuscator consists of a convolutional layer (Conv), a LSTM layer, and a fully-connected layer (FC). At step $t$, the convolutional layer takes the embeddings of the past $n_t$ videos as input ($E_t\in\mathbb{R}^{n_t\times 404}$) and outputs a real vector with $m_1$ elements ($\phi^1_t\in\mathbb{R}^{m_1}$). Next, the LSTM layer takes $\phi^1_t$ and the hidden vector at obfuscation step $t-1$ with $m_3$ elements $h_{t-1}\in\mathbb{R}^{m_3}$ as input, and outputs a real vector with $m_2$ elements ($\phi^2_t\in\mathbb{R}^{m_2}$) and the hidden vector $h_t\in\mathbb{R}^{m_3}$ for obfuscation step $t$ ($m_1=m_2=m_3=128$ in our experiments). 
Finally, a linear layer converts $\phi^2_t$ into a real vector with the same dimension as the video embedding. We denote this vector by $e_t\in\mathbb{R}^{404}$ as it represents the target embedding for the obfuscation video. 
Let $E=[e_1,...,e_M]$ denote the embedding vectors of the $M$ obfuscation videos at our disposal. Then, the probability of selecting the $i$-th obfuscation video, $i=1 \ldots M$, is calculated proportionally to the similarity between its embedding and the target embedding after normalizing using a softmax function:

\begin{equation}
\label{eq:actor}
    \pi(a_t=i|s_t)=\frac{e^{\langle e_t,e_i\rangle}}{\sum_{i=1}^{i=M}e^{\langle e_t,e_i\rangle}},
\end{equation}
where $\langle x,y\rangle$ denotes the inner product between $x$ and $y$. Note that we use the on-policy RL algorithm A2C (Advantage Actor and Critic)\cite{openai2017a2c} to train the \Obfuscator (see Appendix \ref{appendix:train_and_test}).

Recall that the \ToolX \Obfuscator is a non-trivial adaptation of Harpo \cite{zhang2021harpo} to YouTube. An important technical difference is that by calculating a target embedding and then selecting an obfuscation item (video in case of YouTube) based on 
the similarity between its embedding and the target embedding, the \ToolX \Obfuscator can handle an unlimited and varying number of possible obfuscation videos without requiring re-training when the set of obfuscation videos changes.

\subsection{Denoiser Model}
\label{subsubsec:deno_model}
As mentioned in Section \ref{subsec:sys_model}, the \Denoiser has three inputs: the non-obfuscated user persona $V^u$, the obfuscated user persona $V^o$, and the recommended video class distribution of obfuscated user persona $C^o$. The \Denoiser uses two LSTM layers and an FC layer to encode inputs, as shown in Figure \ref{fig:stru3}. Specifically, the first LSTM layer takes as input the embeddings of videos in the non-obfuscated user persona $V^u$ recurrently and outputs its final hidden vector $f^1\in\mathbb{R}^{n}$ (we use $n=128$ in our experiments). Similarly, the inputs of the second LSTM layer are the embeddings of videos in the obfuscated user persona $V^o$ and its output is its last hidden vector $f^2\in\mathbb{R}^{n}$. Last, the FC layer converts the class distribution $C^o\in\mathbb{R}^K$ (where $K$ represents the number of categories) into a real vector $f^3\in\mathbb{R}^{n}$. By concatenating vectors $f^1$, $f^2$, and $f^3$ into a single vector with dimension $3n$, a final FC layer is used to map it into the estimated recommended video class distribution  $\hat C^u\in\mathbb{R}^K$. Note that we train the \Denoiser based on supervised learning with stochastic gradient descent (see Appendix \ref{appendix:train_and_test}).

\subsection{Repopulating Recommended Videos} 
\label{subsubsec:repo_video}
Recall that the \Denoiser in \ToolX outputs a target video class distribution $\hat C^u$.
%
In order to go from a target video class distribution back to actual videos on the user's screen, we repopulate the recommendations using a browser extension.

For efficiency, we maintain a ``bank" of videos per class
and use it to repopulate the recommendations.
This leads to the question of how often should we refresh this bank in order to get a suitable trade-off between the recency of the videos and the overhead required to collect the videos.
To ascertain what the optimal time period would be to refresh this bank we run a 24-hour experiment where we query the name of a class in the YouTube search bar as a proxy for the explicit class and collect statistics for each class's most popular recommended videos.
Specifically, we run the same query each hour, collect the top 20 search results per query, and compute the percentile of top queries that remain the same.
The results indicate that for most classes about 70-80\% of the top search results remain the same. 
Motivated by this, we periodically -- or on an on-demand basis -- crawl a sufficiently large number of videos for each class to re-populate our bank.
%
%
Note that the ``noisy" recommended videos removed during the repopulation process will be included into our obfuscation video sets such that they can be played later to augment the obfuscation effect.

\subsection{YouTube Surrogate Model} 
\label{subsubsec:surr_model}
The training of the \Obfuscator requires frequent interactions with the YouTube recommendation system. However, directly interacting with YouTube is time-consuming, since it takes more than 30 minutes to construct a single persona (as described in Section \ref{subsec:dataset}). To train the \Obfuscator efficiently, we build a surrogate model as a replication of the actual YouTube recommendation system. 

The architecture of our surrogate model consists of a LSTM layer and a FC layer. The LSTM layer takes as input the embeddings of videos in a user persona recurrently and outputs its last hidden vector, which will be used as the input of the FC layer. Then, the FC layer will output the recommended video class distribution $C^i\in\mathbb{R}^K$, where $i\in\{u,o\}$ (see Figure \ref{fig:stru4} in Appendix \ref{appendix:surro_model} for details). Note that we train the surrogate model via supervised learning with stochastic gradient descent (see Appendix \ref{appendix:train_and_test}).
We also provide detailed discussion about the rationale of designing such surrogate model and the differences between our surrogate model and prior works in Appendix \ref{appendix:surro_model}.

\subsection{\ToolX Implementation} 
\label{subsubsec:sys_imp}
We implement \ToolX as a browser extension, which consists of two components: \Obfuscator and \Denoiser.

\paratitle{Obfuscator.} 
The \Obfuscator is a lightly modified version of Harpo's browser extension \cite{zhang2021harpo}.
The browser extension plays the selected obfuscation videos in a background tab that is hidden from users. 
In order to determine the timing of playing obfuscation videos, the \Obfuscator component uses a background script to keep monitoring the URLs visited by the user and estimating the arrival rate of YouTube videos watched by user as $\lambda^u$. 
Then, given obfuscation budget $\alpha$, the \Obfuscator component will use a Poisson Process with rate 
$\lambda^o=\frac{\lambda^u\alpha}{1-\alpha}$ to inject randomly select obfuscation videos. 
%
%
%
To mimic a typical user who watches one video at a time, the selected obfuscation videos can be played only when the user is not already using YouTube. 
However, if a user continues to watch YouTube videos for an extended time period, we can simultaneously play the selected obfuscation videos (in the background as explained above) to prevent YouTube from getting unfettered user watch history.\footnote{It is not entirely uncommon for YouTube users to play videos in multiple browser tabs.} 

\paratitle{Denoiser.} 
The \Denoiser has two modules: HTML modification and the denoising. 
The HTML modification module is implemented in the background script. 
Whenever the user visits YouTube homepage, the HTML modification module sends the ``noisy" homepage recommendation video list requested from the content script to the denoising module. 
Once HTML modification module receives the ``clean" homepage recommended video list from the denoising module, it will modify the HTML of YouTube homepage to show ``clean" homepage recommended videos. 
The denoising module is implemented in the back-end, which is responsible for accessing the metadata of the received ``noisy'' homepage recommended videos, running the \Denoiser model to convert the ``noisy'' homepage recommended video list into a ``clean" one, and then sends the ``clean'' video list back to the HTML modification module. 
We evaluate the implementation overhead of the \Obfuscator and \Denoiser components in Section \ref{subsec:overhead}.

%

\vspace{-.05in}
\section{Experimental Setup}
\label{sec:setup}
\subsection{User Personas}
\label{subsec:persona}

To train and evaluate \ToolX, we need to construct realistic user personas. However, it is challenging to have access to real-world YouTube users' video watch history in a large scale as our training data. To address this concern, we design two approaches that can generate a large number of synthetic user personas to simulate real-world users: 1) the first approach creates sock puppet based personas by following the ``up next'' videos recommended by YouTube; 2) the second approach leverages the YouTube videos publicly posted by Reddit users as an approximation of their YouTube user personas. We use these synthetic user persona datasets to train \ToolX. Then, we evaluate it on both synthetic user persona datasets and a real user persona dataset that contains YouTube video watch history collected from real-world users. We describe these three datasets in detail below.

\vspace{-.03in}
\paratitle{Sock Puppet Based Personas.} According to YouTube, about 70\% of the videos viewed on the platform are sourced from its recommendation system \cite{solsman2021youtubeai}. Accordingly, given the current video, the ``up next'' videos recommended by YouTube are good representations of the potential subsequent videos watched by real-world YouTube users.
Based on this insight, we build a sock puppet user persona model that generates random \textit{recommendation trails} from a single \textit{seed video} to model realistic YouTube user personas, by keeping playing one of the ``up next'' videos recommended by YouTube randomly with uniform probability (see Appendix \ref{appendix:sock_puppet} for details). Since these personas are synthetically built, we are able to exercise more control over the distribution of watched videos. 
%

In total, we generate 10,000 sock puppet based personas with 40 videos each. Note that we set the length of each user persona as 40, since we empirically observe that 40 videos can trigger enough personalized recommended videos on the YouTube homepage and the average time it takes to watch them is close to the average daily time spent by each YouTube user (35 min) \cite{URL_YOUTUBE_BLOG}.

\vspace{-.03in}
\paratitle{Reddit User Personas.} 
As a second way of simulating real-world user personas in a large scale, 
we gather YouTube links publicly posted by social media users as an approximation of their YouTube personas. While there are various social media platforms where users can share YouTube videos, we choose to collect data from Reddit, since it is one of the largest and most popular communities where users post links related to their interests, and millions of Reddit's user submissions\footnote{A Reddit user submission is a json file storing metadata of a Reddit user's posts, including the username, the timestamp, the URL of post, the text, etc.} are publicly available. 

Specifically, we download Reddit user submissions from 2017 to 2021 using APIs provided by \texttt{pushshift.io} \cite{URL_PUSHIFTIO}. For each user submission, we first extract the username and all YouTube links posted by this Reddit account. Next, we filter out any duplicate or broken links. Then, we extract the YouTube video ids from these remaining links in order.
Finally, we remove users with less than 40 YouTube video posts, since a small number of videos is not a fair approximation of the user's actual YouTube persona. In total, we collect 10,000 Reddit user personas with length 40.

\vspace{-.03in}
\paratitle{Real-world YouTube Users.} 
To conduct a more realistic evaluation of \ToolX, we use a real-world user dataset from \cite{casas2022exposure}. This dataset contains the web browsing histories of 936 real users collected through Web Historian \cite{URL_WEBHistorian} for three months. It is a good representative of real YouTube users, since: 1) the demographic distribution of these users, including their gender, age (18-65+), and education level (from less than high school to Doctoral degree), are relatively uniform; 2) on average 650 YouTube video URLs are watched by each user in three months; 3) the first 40 videos watched by these users have different video class distribution, indicating diverse user interests. Considering that the dataset is collected over a long period, we select the first 40 YouTube videos watched by each of these 936 users as our real user personas, to evaluate \ToolX.

\vspace{-.03in}
\subsection{Data Collection}
\vspace{-.03in}
\label{subsec:dataset}
\paratitle{User Persona Construction.} We use a fresh Firefox browser based on Selenium to construct each user persona. For each sock puppet based persona, we start with a seed video and then follow the ``up next" video recommendations to generate a \textit{recommendation trail}. We play each video in a user persona for 30 seconds before playing the next video. Note that we clear any pop-up windows and skip the ads before playing the video. For each Reddit user and real user persona, since we already known the video ids in each persona, we visit these videos sequentially\footnote{Note that directly visiting the URL of each video doesn't trigger cookies from YouTube and hence no personal recommendation can happen. To address this, we first search the video id at YouTube and then click the first search result.}. Similar to constructing synthetic user personas, if there are any pop-up windows or ads, we clear them and then play the video for 30 seconds. 

\vspace{-.03in}
\paratitle{Recommended Video Collection.} After we complete the construction of each user persona, we go back to the YouTube homepage and refresh it for 50 times to collect all the recommended videos into a list. Note that we refresh the homepage multiple times since we want to collect enough homepage recommended videos to estimate the recommended video class distribution. We choose the number of refresh times as 50 since we empirically observe that it is a good tradeoff between collecting enough samples and minimizing the quantity of crawls to be performed. Because extremely popular videos are common across many users regardless of their profile, we 
remove them to underscore personalized recommendations. With this in mind, we filter out videos which appear in more than 1\% of personas' homepage recommended video lists. We also exclude YouTube videos showing in the homepage of a fresh browser.
Then, for each recommended video, we extract the associated tags (i.e. a list of keywords) from its metadata, and map each of them into one of the 154 topic-level video classes we have (note that a video may belong to multiple video classes).
Last, for each persona, we count the number of recommended videos in each class and divide it by the sum of videos in all classes to derive the recommended video class distribution of each persona.

\subsection{Training and Testing}
\label{subsec:train_test}
We discuss details about how we prepare the training and testing datasets, and use them to train and test  surrogate model, \Obfuscator and \Denoiser in Appendix \ref{appendix:dataset}-\ref{appendix:train_and_test}.

\subsection{Baselines}
\paratitle{Obfuscator.}
We compare the privacy-enhancing performance of \ToolX \Obfuscator with three baselines:

\textit{1) Rand-Obf:} At each obfuscation step, we randomly select one obfuscation video from the obfuscation video set, and the probability of selecting each obfuscation video is equal to $\frac{1}{M}$ ($M$ is the total number of obfuscation videos in the set).

\textit{2) Bias-Obf:} At each obfuscation step, we randomly select one obfuscation video from the obfuscation video set. However, the probability of selecting each obfuscation video is proportional to the reward triggered by each obfuscation video. To create such non-uniform distribution, we first use Rand-Obf to randomly select obfuscation videos and then record the reward after injecting them into non-obfuscated user personas. We repeat this experiment for 50 epochs and count the accumulative reward of each obfuscation video, normalize it by the sum of the accumulative rewards of all obfuscation videos, and use the normalized rewards as the non-uniform probability distribution.

\textit{3) PBooster-Obf:} At each obfuscation step, we select one obfuscation video from the obfuscation video set which can maximize the reward for the current step based on the greedy algorithm PBooster proposed in \cite{beigi2019protecting}.

\vspace{-.03in}
\paratitle{Denoiser.}
We compare the utility-preserving performance of the \Denoiser in \ToolX with a baseline that uses the same architecture as the surrogate model to predict $C^u$ directly from a non-obfuscated user persona $V^u$, without taking the obfuscated persona $V^o$ and the associated recommended video class distribution $C^o$ as inputs. We refer to this baseline as \textit{Surro-Den}. Ideally, if the surrogate model is a perfect replication of YouTube's recommendation system, then users could directly use it to get recommended videos based on their non-obfuscated user personas. Clearly this is unrealistic in practice since the surrogate model does not have access to the complete universe of YouTube videos which are updated constantly, and the model is merely an approximation of the actual YouTube recommendation system. 

For convenience, we denote the \ToolX \Obfuscator and \ToolX \Denoiser by \ToolX-Obf and \ToolX-Den respectively in the rest of the paper.
\section{Evaluation}
In this section, we evaluate the effectiveness of \ToolX from six perspectives: privacy, utility, overhead, stealthiness, robustness to de-obfuscation, and personalization.

\subsection{Privacy}
\label{subsec:privacy}
We first evaluate the effectiveness of \ToolX in enhancing privacy using three user persona datasets, and report the results in TABLE \ref{tab:privacy}. Note that we test \ToolX-Obf  and other obfuscator baselines against the real-world YouTube recommendation system. 

As shown in TABLE \ref{tab:privacy_1}, \ToolX-Obf can trigger 0.91 KL divergence in the recommended video class distribution after obfuscation ($P$) on sock puppet based personas, which translates into triggering 41.63\% of the maximum possible KL divergence in the recommended video class distribution ($P^{Norm}$). Compared with other baselines, \ToolX-Obf can increase $P^{Norm}$ by up to 2.01$\times$ and at least 1.33$\times$. Similarly, on Reddit user personas, \ToolX-Obf outperforms all baselines by up to 1.57$\times$ and at least 1.32$\times$, as reported in TABLE \ref{tab:privacy_2}.

Moreover, we evaluate whether the effectiveness of \ToolX in enhancing privacy can be transferred to real-world user personas. Specifically, we use the same obfuscator trained on sock puppet based personas to inject obfuscated videos into real-world user's video watch history, and then test it against YouTube. As reported in TABLE \ref{tab:privacy_3}, \ToolX-Obf can trigger 87.23\% of the maximum possible KL divergence in the recommended video class distribution ($P^{Norm}$), which outperforms all baselines against YouTube by up to 1.92$\times$ and at least 1.58$\times$ in terms of $P^{Norm}$.

\renewcommand{\arraystretch}{1.0}
\begin{table}[t]
\footnotesize
\vspace{-.05in}
\caption{Privacy evaluation results against YouTube w.r.t. $P$ and $P^{Norm}$.}
\vspace{-.1in}
\begin{subtable}{0.45\textwidth}
\centering
\begin{tabular}{p{0.6in}<{\centering}p{0.6in}<{\centering}p{0.6in}<{\centering}p{0.6in}<{\centering}p{0.8in}<{\centering}}
\hline
Obfuscator & Rand-Obf  & Bias-Obf & PBooster-Obf & \textbf{\ToolX-Obf} \\\hline
$P$                 & 0.71      & 0.70   & 0.81    & \textbf{0.91} \\
$P^{Norm}$          & 21.55\%   & 20.76\% & 31.24\%   & \textbf{41.63\%} \\\hline
\end{tabular}
\vspace{-.06in}
\caption{Using sock puppet based personas ($D^{Min}\!:\!0.49, D^{Max}\!:\!1.51$).}
\label{tab:privacy_1}
\end{subtable}
\vspace{.02in}
\begin{subtable}{0.45\textwidth}
\centering
\begin{tabular}{p{0.6in}<{\centering}p{0.6in}<{\centering}p{0.6in}<{\centering}p{0.6in}<{\centering}p{0.8in}<{\centering}}
\hline
Obfuscator & Rand-Obf  & Bias-Obf & PBooster-Obf & \textbf{\ToolX-Obf} \\\hline
$P$                 & 1.05      & 1.07  & 1.13     & \textbf{1.30} \\
$P^{Norm}$          & 48.79\%    & 50.99\%  & 57.84\%    & \textbf{76.49\%} \\\hline
\end{tabular}
\vspace{-.06in}
\caption{Using Reddit user personas ($D^{Min}\!:\!0.60, D^{Max}\!:\!1.51$).}
\label{tab:privacy_2}
\end{subtable}
\vspace{.02in}
\begin{subtable}{0.45\textwidth}
\centering
\begin{tabular}{p{0.6in}<{\centering}p{0.6in}<{\centering}p{0.6in}<{\centering}p{0.6in}<{\centering}p{0.8in}<{\centering}}
\hline
Obfuscator & Rand-Obf  & Bias-Obf & PBooster-Obf & \textbf{\ToolX-Obf} \\\hline
$P$                 & 0.98      & 1.00  & 1.05    & \textbf{1.39} \\
$P^{Norm}$          & 45.45\%   & 48.01\% & 55.34\%   & \textbf{87.23\%} \\\hline
\end{tabular}
\vspace{-.06in}
\caption{Using real-world user personas ($D^{Min}\!:\!0.53, D^{Max}\!:\!1.51$).}
\label{tab:privacy_3}
\end{subtable}
\label{tab:privacy}
\vspace{-.2in}
\end{table}

\subsection{Utility}
\label{subsec:utility}
Next, we evaluate the effectiveness of \ToolX in preserving user utility. 
TABLE \ref{tab:utility_1} reports our evaluation results  in terms of $U_{Loss}$ and $U^{Norm}_{Gain}$ using sock puppet based personas. 
Compared with Surro-Den, \ToolX-Den achieves on average 26\% better performance in terms of decreasing $U_{Loss}$ (i.e. increasing $U^{Norm}_{Gain}$). Recall that different from Surro-Den, \ToolX-Den also takes as inputs the obfuscated user persona $V^o$, and the associated recommended video class distribution $C^o$ which comes directly from the actual YouTube system. In contrast, the surrogate model is merely a ``first-order" model of the actual, quite complex YouTube system.
%
We also evaluate the effectiveness of \ToolX-Den in preserving user utility using both Reddit user personas and real-world users. As reported in TABLE \ref{tab:utility_2}-\ref{tab:utility_3}, \ToolX-Den can consistently preserve the utility well, reducing the utility loss by 93.80\% and 90.40\% respectively.

It is worth noting that the effectiveness of the denoiser in preserving utility does not depend on the effectiveness of the obfuscator in enhancing privacy. As shown in Table \ref{tab:utility_1}-\ref{tab:utility_3}, the same denoiser can achieve almost the same utility loss $U_{Loss}$ under different obfuscators, which implies the denoiser does not need to sacrifice privacy in order to preserve utility. We discuss the privacy-utility tradeoff in the next subsection.

\begin{table}[t]
\footnotesize
\vspace{-.05in}
\caption{Utility evaluation results w.r.t. $U_{Loss}$ and $U^{Norm}_{Gain}$. Note that each cell in the table reports $U_{Loss}/U^{Norm}_{Gain}$. }
\vspace{-.08in}
\begin{subtable}{0.45\textwidth}
\centering
\begin{tabular}{p{1.3in}<{\centering}p{1.0in}<{\centering}p{1.0in}<{\centering}}
\hline
\diagbox[width=0.8in]{Obfuscator}{Denoiser} & Surro-Den & \textbf{\ToolX-Den} \\\hline
Rand-Obf            & 0.60 / 50.91\%  & 0.54 / 79.09\%   \\
Bias-Obf            & 0.60 / 49.06\%  & 0.53 / 82.08\%    \\
PBooster-Obf   & 0.60 / 66.14\%  & 0.53 / 86.83\%    \\
\textbf{\ToolX-Obf}     & 0.60 / 74.59\%   & \textbf{0.53 / 90.35\%} \\\hline
\end{tabular}
\vspace{-.06in}
\caption{Using sock puppet based personas ($D^{Min}\!:\!0.49$).}
\label{tab:utility_1}
\end{subtable}

\begin{subtable}{0.45\textwidth}
\centering
\begin{tabular}{p{1.3in}<{\centering}p{1.0in}<{\centering}p{1.0in}<{\centering}}
\hline
\diagbox[width=0.8in]{Obfuscator}{Denoiser} & Surro-Den & \textbf{\ToolX-Den} \\\hline
Rand-Obf            & 0.68 / 83.26\%   & 0.64 / 91.18\% \\
Bias-Obf            & 0.68 / 83.98\%   & 0.66 / 88.96\% \\
PBooster-Obf   & 0.68 / 85.88\%  & 0.65 / 90.46\%    \\
\textbf{\ToolX-Obf}     & 0.68 / 89.32\%   & \textbf{0.65 / 93.80\%} \\\hline
\end{tabular}
\vspace{-.06in}
\caption{Using Reddit user personas ($D^{Min}\!:\!0.6$).}
\label{tab:utility_2}
\end{subtable}

\begin{subtable}{0.45\textwidth}
\centering
\begin{tabular}{p{1.3in}<{\centering}p{1.0in}<{\centering}p{1.0in}<{\centering}}
\hline
\diagbox[width=0.8in]{Obfuscator}{Denoiser} & Surro-Den & \textbf{\ToolX-Den} \\\hline
Rand-Obf            & 0.66 / 70.79\%   & 0.62 / 81.12\% \\
Bias-Obf            & 0.66 / 72.34\%   & 0.61 / 82.34\% \\
PBooster-Obf   & 0.66 / 76.99\%  & 0.61 / 85.31\%    \\
\textbf{\ToolX-Obf}     & 0.66 / 84.78\%   & \textbf{0.61 / 90.40\%} \\\hline
\end{tabular}
\vspace{-.06in}
\caption{Using real-world user personas ($D^{Min}\!:\!0.53$).}
\label{tab:utility_3}
\end{subtable}
\label{tab:utility}
\vspace{-.20in}
\end{table}

\vspace{-.03in}
\subsection{Varying the Obfuscation Budget}
\label{subsec:tradeoff}
So far, the obfuscation budget $\alpha$ is set to 0.2 in our evaluation. To evaluate how the obfuscation budget (i.e. the percentile of obfuscation videos in a user persona) can affect the performance of \ToolX, we increase the value of $\alpha$ and evaluate how the performance of \ToolX changes w.r.t. both privacy ($P^{Norm}$) and utility ($U_{Loss}$). We use sock puppet based persona dataset and consider three baselines:  Rand-Obf/\ToolX-Den (i.e. the combination of Rand-Obf and the \ToolX \Denoiser), Bias-Obf/\ToolX-Den (i.e. the combination of Bias-Obf and the \ToolX \Denoiser), and PBooster-Obf/\ToolX-Den (i.e. the combination of PBooster-Obf and the \ToolX \Denoiser).

\paratitle{Privacy-utility tradeoff.}
Figure \ref{fig:pu1} shows the privacy-utility tradeoff between $P^{Norm}$ and $U_{Loss}$ with varying $\alpha$ from $\{0.2, 0.3, 0.5\}$, where the top left region corresponds to both high privacy and utility. We observe that, with \ToolX-Den, the utility loss caused by different obfuscators can be significantly reduced without sacrificing privacy. Note that since our denoiser is designed to work after obfuscation, it does not hurt the performance of the obfuscator.
Moreover, with \ToolX-Den, the utility loss  remains almost the same as we keep increasing the obfuscation budget to get higher privacy. For example, compared with baselines without using \ToolX-Den, \ToolX can reduce the utility loss by 2.12$\times$ when $\alpha=0.5$. Note that without \ToolX-Den, the obfuscator needs to sacrifice utility (higher utility loss) to achieve higher privacy. This is a key difference between \ToolX and prior works that consider the privacy-utility tradeoff (see Section \ref{sec:related}).

\paratitle{Obfuscation budget and privacy level.} 
Recall that we use the recommended video class distribution as a proxy to a user profile, see Section \ref{subsec:goal}. To evaluate whether \ToolX can privatize a user profile to look almost random, we increase the obfuscation budget beyond 0.5 aiming to achieve a $P^{Norm}$ value as close to 100\% as possible.
As shown in Figure \ref{fig:pu2}, for $\alpha$ equal to 0.7 (i.e. 70\% of the videos in a user persona are obfuscation videos), $P^{Norm}$ will reach 92.95\%, which means the on-average (averaged over all users) divergence between the recommended video class distribution before and after obfuscation is 93\% of the on-average divergence between the recommended video class distribution of two random users. It is also worth noting that for real-world user personas, $P^{Norm}$ can get very close to 100\% with $\alpha$ merely equal to 0.5. Hence, we conclude that \ToolX can achieve meaningful privacy for practical obfuscation budgets $\alpha$.



\begin{figure}[!h]
\centering
\includegraphics[width=.75\linewidth]{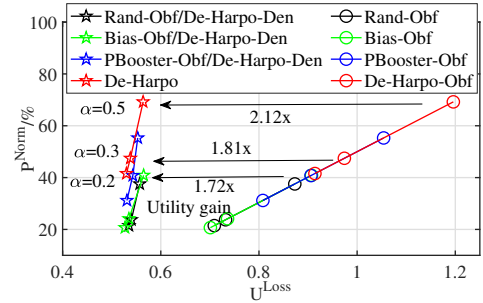}
\vspace{-.1in}
\caption{Privacy-utility tradeoff w.r.t. $P^{Norm}$ and $U_{Loss}$ under different obfuscation budget $\alpha$. Note that Rand-Obf/\ToolX-Den represents the combination of Rand-Obf obfuscator and the \ToolX \Denoiser, Bias-Obf/\ToolX-Den represents the combination of Bias-Obf obfuscator and the \ToolX \Denoiser, and PBooster-Obf/\ToolX-Den represents the combination of PBooster-Obf obfuscator and the \ToolX \Denoiser. Top left of figure represent both high privacy and high utility.}
\label{fig:pu1}
\vspace{-.15in}
\end{figure}


\vspace{-.05in}
\subsection{Overhead}
\label{subsec:overhead}
\paratitle{Obfuscation budget and overhead.} 
The larger the obfuscation budget the larger the overhead as more obfuscation videos need to be injected in the video watch history. Not surprisingly, as shown in Figure \ref{fig:pu1}, with increasing obfuscation budget $\alpha$, the privacy ($P^{Norm}$) will increase for all obfuscators. That said, \ToolX can increase privacy with less obfuscation budget than the rest. Specifically, with $\alpha=0.2$, \ToolX can achieve the same level of privacy as other baselines achieve with $\alpha=0.5$. That is, \ToolX can be as effective as baseline obfuscator in terms of enhancing privacy with 2.5$\times$ less obfuscation budget. 

\paratitle{System overhead.} We evaluate the system overhead of \ToolX in terms of CPU and memory usage and the video page load time using a an Intel i7 workstation with 64GB RAM on a campus WiFi network. 
We report that for the \Obfuscator component, the increased CPU and memory usage are less than 5\% and 2\% respectively, and the increased video page load time is less than 2\% even when $\alpha=0.5$. For the \Denoiser component, the increased CPU and memory usage are less than 28\% and 3\% respectively, and the YouTube's homepage load time is only increased by 38 millisecond. Overall, we conclude that \ToolX has a negligible impact on the user experience.
(See Appendix \ref{appendix:sys_overhead} for more detailed analysis).

\begin{table}[t]
\centering
\footnotesize
\caption{Stealthiness evaluation results under different obfuscation budget $\alpha$ with 5\% DeHarpo users. Note that we choose $\alpha$ from $\{0.2,0.3,0.5\}$ and report (Precision, Recall) of the adversarial detector for different obfuscators.}
\vspace{-.05in}
\begin{tabular}{p{1.0in}<{\centering}p{0.8in}<{\centering}p{0.8in}<{\centering}
p{0.8in}<{\centering}}
\hline
\multirow{2}{*}{Obfuscator} & \multicolumn{3}{c}{(Precision, Recall)} \\ \cline{2-4}
 & $\alpha=0.2$    & $\alpha=0.3$         & $\alpha=0.5$ \\\hline
Rand-Obf             & (4\%, 99\%)         & (5\%, 92\%)            & (5\%, 92\%) \\
Bias-Obf             & (7\%, 72\%)         & (5\%, 81\%)            & (36\%, 94\%) \\
PBooster-Obf & (19\%, 86\%) & (16\%, 93\%) & (49\%, 88\%) \\
\textbf{\ToolX-Obf}     & (67\%, 98\%)         & (73\%, 99\%)            & (74\%, 99\%) \\\hline
\end{tabular}
\label{tab:stealth}
\vspace{-.1in}
\end{table}


\begin{figure}[!t]
\begin{minipage}[t]{.23\textwidth}
    \includegraphics[width=.99\linewidth]{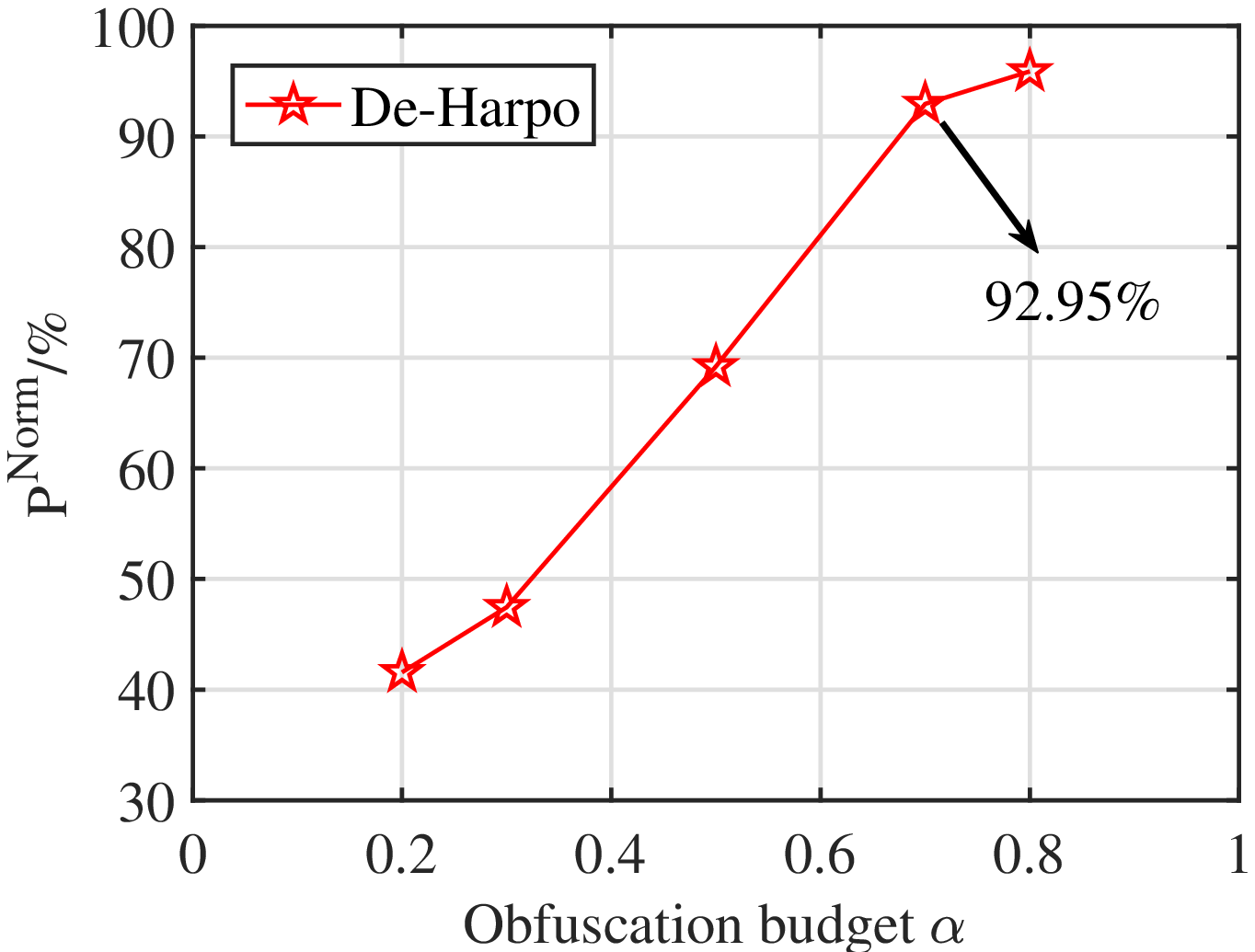}
    \caption{Privacy level $P^{Norm}$ vs obfuscation budget $\alpha$.}
    \vspace{-.1in}
    \label{fig:pu2}
\end{minipage}
\begin{minipage}[t]{.23\textwidth}
    \includegraphics[width=.99\linewidth]{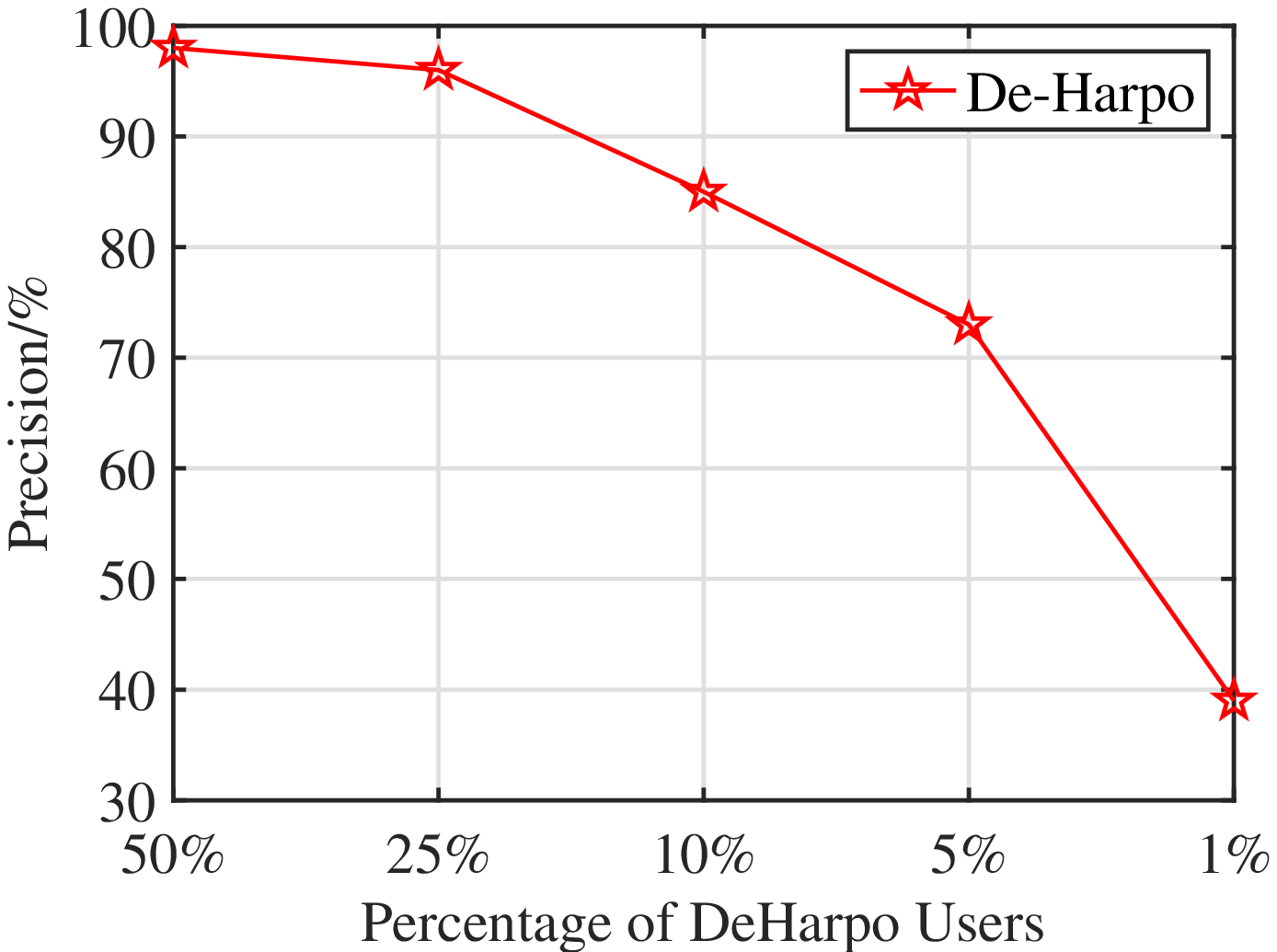}
    \caption{Precision of the adversarial detector vs the percentage of \ToolX users under $\alpha=0.5$.}
    \vspace{-.1in}
    \label{fig:stealthy}
\end{minipage}
\label{fig:}
\vspace{-.2in}
\end{figure}

\subsection{Stealthiness}
\label{subsec:stealthy}
In this subsection, we evaluate whether an adversary can train an ML classifier to accurately detect the usage of obfuscators.
We use the precision and recall of this adversarial detector to measure stealthiness of obfuscation. 
If the detector achieves high precision and recall, then it means that an obfuscator is less stealthy. 
Specifically, the input of the adversary is a user persona consisting of a sequence of videos and the binary output indicates whether or not the user persona contains at least one obfuscation video.

We train the adversarial classifier via supervised learning. 
To create the labeled dataset, we use the same set of non-obfuscated and obfuscated sock puppet based user personas used for evaluation in Section \ref{subsec:overhead} as inputs, and assign the corresponding labels to the personas (0: non-obfuscated, 1: obfuscated).
%
For each obfuscator and obfuscation budget $\alpha$, we get a balanced training dataset with 1,440 obfuscated personas and the corresponding 1,440 non-obfuscated personas. After training, we use an unbalanced dataset with 5\% obfuscation personas (a total of 360 obfuscated personas and 6,840 non-obfuscated personas) to test the detector, since only a small fraction of YouTube users are expected to employ \ToolX.

Table \ref{tab:stealth} reports the testing precision and recall of the adversarial detector under different $\alpha$ values. 
We observe that as $\alpha$ increases, both the precision and recall of the detector also increase.
This is expected as larger $\alpha$ represents more obfuscation videos, which makes it easier for the adversarial detector to distinguish obfuscated personas from non-obfuscated personas. 
%

Not surprisingly, Rand-Obf is the most stealthy obfuscator since it injects obfuscation videos randomly. \ToolX-Obf, which injects obfuscation videos that introduce new user interests to confuse YouTube, can still achieve reasonable stealthiness even when $\alpha=0.5$. 
%
Specifically, it leads to 74\% precision (36\% false positive rate) even with $\alpha=0.5$. Note the the high false positive rate presents a major obstacle in deployment of the adversarial detector due to base-rate fallacy \cite{axelsson2000base}.
We further vary the percentage of \ToolX users over all YouTube users to show how the precision of the adversarial detector changes as we go from a very unbalanced dataset to a perfectly balanced one. As shown in Figure \ref{fig:stealthy}, as the  percentage of \ToolX users varies from 1\% to 50\%, the adversarial detector's precision will increase, as expected. However, it is unlikely in practice that a large fraction of YouTube users will use obfuscation measures. And, even in the case of a balanced dataset, a 2\% false positive rate still corresponds to tens of millions of users making it prohibitively expensive to deploy it.
%
%
Essentially, the adversarial detector will have to achieve exceptionally high precision to be useful in practice.

Note that such a binary detector may be used as a first step of the detection workflow. 
Once the adversary detects the usage of \ToolX, it may further attempt to de-obfuscate the obfuscated user personas. 
That is, the adversary may attempt to identify obfuscation videos in the obfuscated user persona such that it may remove them to retrieve the non-obfuscated user personas. 
We evaluate this de-obfuscation performance of an adversary next.

\begin{table}[t]
\centering
\footnotesize
\caption{De-obfuscation robustness evaluation results under different obfuscation budget. Note that we set $\alpha\in\{0.2,0.3,0.5\}$ and report (Precision, Recall) of adversarial detector under different obfuscation approaches.}
\vspace{-.05in}
\begin{tabular}{p{1.0in}<{\centering}p{0.8in}<{\centering}p{0.8in}<{\centering}
p{0.8in}<{\centering}}
\hline
\multirow{2}{*}{Obfuscator} & \multicolumn{3}{c}{(Precision, Recall)} \\ \cline{2-4}
                    & $\alpha=0.2$    & $\alpha=0.3$         & $\alpha=0.5$ \\\hline
Rand-Obf             & (62\%, 97\%)         & (67\%, 91\%)            & (69\%, 99\%) \\
Bias-Obf             & (67\%, 89\%)         & (71\%, 89\%)            & (77\%, 92\%) \\
PBooster-Obf   & (68\%, 93\%) & (71\%, 90\%) & (77\%, 94\%) \\
\textbf{\ToolX-Obf}     & (79\%, 93\%)         & (83\%, 97\%)            & (84\%, 97\%) \\\hline
\end{tabular}
\label{tab:robust}
\vspace{-.15in}
\end{table}

\subsection{De-obfuscation Robustness}
\label{subsec:robust}
Once the adversary detects the usage of \ToolX in a user persona, it can conduct de-obfuscation. 
To evaluate whether an obfuscator is robust to de-obfuscation, we train a second adversarial detector to distinguish the obfuscation videos from the actual user videos. 
Specifically, we build a second ML classifier to detect the type of each video (user versus obfuscation video) in each sock puppet based user persona, and use its precision and recall to measure the de-obfuscation robustness. 
Smaller precision and recall represents higher de-obfuscation robustness.

We use the same set of obfuscated personas as in Section \ref{subsec:overhead} as inputs. 
For each video in an obfuscated user persona, we assign a binary label, where 0 represents it is watched by the user while 1 represents that it is injected by the obfuscator. 
The detector model takes as input the obfuscated user persona, and predicts a label for each video in the user persona. 
We use a recurrent neural network (LSTM layer) to model this adversarial detector.

As shown in Table \ref{tab:robust}, 
the precision of this adversarial detector is lower than 85\%, which means more than 15\% of the obfuscated videos identified by the adversary are false positives (they are actual user videos).
Similar to stealthiness, false positives present a bigger challenge to the adversary in deploying this detector in practice. 
Hence, we conclude that \ToolX is robust to de-obfuscation by an adversary. 

Note that while the adversary has lower precision against Rand-Obf and Bias-Obf than agaisnt \ToolX, this is because \ToolX is 2.5$\times$ more effective in preserving privacy (see Section \ref{subsec:overhead}), thus, overall, it is more privacy-preserving.

\subsection{Personalization}
\ToolX so far is trained to maximize the KL divergence in the recommended video class distribution after obfuscation, by either increasing or reducing the probability of each video class. However, a YouTube user may have a list of \textit{sensitive} video classes (e.g. health or wellness related), where they do not want the YouTube recommendations containing videos from these classes after obfuscation (i.e. reducing their probability to zero). 

Motivated by this, we design a mechanism that can treat \textit{sensitive} video classes and \textit{non-sensitive} video classes differently based on user preferences. Without loss of generality, suppose the first $L$ classes of the recommended video class distribution are \textit{non-sensitive} and the remaining $K-L$ classes are \textit{sensitive}. We then train \ToolX to maximize the following privacy metric, which aims to treat non-sensitive classes like before (maximize divergence before and after obfuscation) and eliminate sensitive class videos:

\begin{equation}
    P^{Personalized} \!\!= \!E[\underbrace{D_{KL}(C^o_{1:L},\!C^u_{1:L})}_{\mbox{\scriptsize $D_{KL}^{NonSens}$}}-\!\lambda \underbrace{D_{KL}(C^o_{L+1:K},\![\epsilon]_{L+1:K})}_{\mbox{\scriptsize $D_{KL}^{Sens}$}}],
\end{equation}
where $[\epsilon]_{L+1:K}\in\mathbb{R}^{K-L}$ indicates a 
close-to-zero vector filled with a
small positive number $\epsilon$ (e.g. $0.0001$), and $\lambda>0$ is an adjustable parameter for controlling the relative importance of $D_{KL}^{NonSens}$ versus $D_{KL}^{Sens}$. Specifically, the term $D_{KL}^{NonSens}$ aims to maximize the distance between the distribution of \textit{non-sensitive} classes before and after obfuscation, like we did before for all classes. The term $-\lambda D_{KL}^{Sens}$ aims to minimize the distance between the distribution of the sensitive classes and a distribution of very small probabilities.\footnote{Notice that we are somewhat abusing the ``distribution" term above, because we do not re-normalize the corresponding probabilities to sum up to 1, as this would (i) de-emphasize the contrast between the patterns of interest and (ii) is not required to meaningfully use the KL divergence formula.} 

Table \ref{tab:personalization} reports our evaluation results of personalized \ToolX against surrogate models, where we select 27 out of 154 video classes related to Beauty \& Wellness and Sports \& Fitness as \textit{sensitive} classes. Compared with non-personalized \ToolX, personalized \ToolX can reduce the divergence between \textit{sensitive} video class distribution and a zero vector ($D_{KL}^{Sens}$) by more than 80\%, while still triggering high divergence in \textit{non-sensitive} class distribution ($D_{KL}^{NonSens}$).

\begin{table}[!t]
\centering
\footnotesize
\caption{Personalization results. $D^{NonSens}_{KL}$ and $D^{Sens}_{KL}$ denote the divergence in \textit{non-sensitive} classes and \textit{sensitive} classes respectively.}
\vspace{-.05in}
\label{tab:personalization}
\begin{tabular}{p{1.3in}<{\centering}
p{1in}<{\centering}p{1in}<{\centering}}
\hline
                                 & $D_{KL}^{NonSen}$            & $D_{KL}^{Sen}$          \\ \hline
\ToolX                           & 1.18           & 0.26          \\
Personalized \ToolX              & 0.81 ($\downarrow 31.36\%$)           & 0.05 ($\downarrow 80.77\%$)       \\\hline
\end{tabular}
\vspace{-.1in}
\end{table}

\section{Discussion}
\label{sec: discussion}
\subsection{Ethical Considerations}
We outline the potential benefits and harms to the user and the recommendation system. 
We argue that the potential benefits of \ToolX outweigh its potential harms.

\noindent\textbf{Users.}
\ToolX provides a clear privacy benefit to its users, especially when platforms such as YouTube do not provide any meaningful control over its tracking and profiling of users.
Crucially, \ToolX is able to enhance privacy while mostly preserving the utility of personalized recommendations. 
Thus, \ToolX does not degrade user experience on YouTube. 
However, users of \ToolX potentially violate YouTube's Terms of Service (TOS) \cite{youtubeTOS} because YouTube might interpret obfuscation as ``fake engagement''. 
Therefore, if a user is signed-in to YouTube, their YouTube account might be suspended if YouTube is able to detect \ToolX's usage (though we showed that YouTube would be unable to do so without risking significant collateral damage). 
More seriously, the violation of TOS might be considered possible violation of the Computer Fraud and Abuse Act (CFAA, 18 U.S. Code § 1030) \cite{cfaa}.
However, given that \ToolX users only watch videos that they are authorized to (i.e., publicly available videos), we argue that the videos injected to the watch history for obfuscation nor the videos injected to the recommendations for repopulation exceed authorized access that could be a violation of CFAA \cite{mackey2021vanburen}.

\noindent\textbf{YouTube.}
Since \ToolX aims to preserve utility of recommendations to YouTube users, we argue that it will not directly hurt user engagement on YouTube.
\ToolX's \Obfuscator and \Denoiser would, however, contribute to additional traffic to YouTube servers and may have some indirect impact on the effectiveness of YouTube's recommendation system,
if a large enough portion of the users adopt De-Harpo. 
%
%
We note that \ToolX can be applied with satisfactory trade-off privacy vs. utility as long as only a minority of YouTube users employ obfuscation tools, which is arguably a realistic expectation. Otherwise, if a significant fraction of users adopts \ToolX, the obfuscation may lead to data poisoning, which will indirectly affect the quality of recommendations for all users. In this case, and in the absence of legal regulation of tracking and user profiling by YouTube, future research will need to explore an alternative scalable solution for privacy preservation that is complementary to obfuscation. Overall, as compared to extant privacy-enhancing obfuscation tools, we conclude that \ToolX is more favorable since it specifically aims to preserve utility and user engagement on YouTube.


\vspace{-0.05in}
\subsection{Limitations \& Future Work}
\vspace{-0.07in}
\paratitle{Side channels.} \ToolX's stealthiness can be undermined by exploiting various implementation side channels. 
For example, YouTube could use Page Visibility API \cite{URL_PAGEVISIBILITYAPI} or the Performance API \cite{URL_CHROMIUMTHROTTLE} to detect whether obfuscation videos are unusually not being played in the foreground.
However, there are patches such as wpears \cite{Anti-Visibility} to avoid detection. Additionally, the \Obfuscator plays the obfuscation videos in full in a background tab while disabling background throttling (or other such optimizations \cite{URL_CHROMIUMTHROTTLE, Non-Active-Tabs}) to prevent detection by such side channels.
As another example, the repopulation of recommendations on the homepage after denoising would entail manipulation of the HTML DOM \cite{HTMLDOM}, which might be detectable. However, such an attach would be infeasible in practice, because the detection approaches would add an overhead of up to several seconds \cite{karami2020carnus, starov2017xhound}.

\vspace{-0.07in}
\paratitle{Deployment on mobile devices.} \ToolX is currently implemented as a browser extension for desktops.
%
Since browser extensions are not supported on iOS or Android, the only option for users to benefit from \ToolX on their mobile phone is to use other Chromium based browsers that allow extensions \cite{Yandex,KiwiBrowser}. 
Another option for mobile users is to use a remote desktop utility \cite{Chrome-remote-desktop} to access YouTube with \ToolX on a desktop. 
Finally, users might still be able to reap the obfuscation benefits of \ToolX if they deploy the extension on their desktop and be logged-in to the same Google account \cite{Google-AllDevices} on both their mobile app and desktop with \ToolX. 
%

\section{Related Work}
\label{sec:related}
In this section, we discuss prior work on privacy-enhancing obfuscation in recommendation systems.

One line of prior research focuses on developing privacy-enhancing obfuscation approaches in online behavioral advertising.   
These efforts are relevant to our work because online behavioral advertising is essentially a recommendation system where the advertiser's goal is to ``recommend'' personalized ads to users based on their online activity. 
However, most of these privacy-enhancing obfuscation approaches are not designed to preserve the utility (i.e., relevance of personalized ads) \cite{howe2017engineering,URL_TRACKTHIS,meng2014your,degeling2018tracking,kim2018adbudgetkiller}, 
as they generally randomly insert a curated set of obfuscation inputs to manipulate online behavioral advertising. 

TrackThis \cite{URL_TRACKTHIS} by Mozilla injects a curated list of 100 URLs to obfuscate a user's browsing profile.
%
%
AdNauseam \cite{howe2017engineering}  clicks a random set of ads to ``confuse'' advertisers.
%
%
One subset of these efforts propose ``pollution attacks'' against online behavioral advertising that  also serve a dual role as privacy-enhancing obfuscation \cite{meng2014your,degeling2018tracking,kim2018adbudgetkiller}.
Meng et al. \cite{meng2014your} propose a pollution attack that can be launched by publishers to increase their advertising revenue by manipulating advertisers into targeting higher paying ads.
The attack involves the addition of curated URLs into a user's browsing profile. 
%
%
%
Degeling et al. \cite{degeling2018tracking} and Kim et al. \cite{kim2018adbudgetkiller} propose similar attacks but focus on two distinct stages of the online behavioral advertising pipeline: user profiling and ad targeting.
Degeling et al. \cite{degeling2018tracking} propose an obfuscation approach that involves adding URLs posted on Reddit into a user's browsing profile. 
%
%
Kim et al. \cite{kim2018adbudgetkiller} propose ``AdbudgetKiller" that involves adding a sequence of URLs into a user's browsing profile to trigger retargeted ads, which are costly for advertisers and waste their advertising budget.

Moving beyond online behavioral advertising, Xing et al \cite{xing2013take} propose pollution attacks against more general personalized recommendation systems such as YouTube, Amazon, and Google Search. 
The authors show that personalized recommendations could easily be manipulated by injecting random or curated obfuscation inputs. 
Since the attack's victim is the user, the work does not take into account the utility of recommendations to the user. 
In contrast, \ToolX is a privacy-enhancing obfuscation system that also takes into account the utility of the recommendations.

Follow up privacy-enhancing obfuscation systems do attempt to take into account the utility-privacy trade-off.
Beigi et al, \cite{beigi2019protecting} propose PBooster, a greedy search approach to obfuscate a user's browsing profile while also keeping utility in consideration.
PBooster employs topic modeling to select a subset of target topics and corresponding obfuscation URLs to add in a user's browsing history. 
Zhang et al. \cite{zhang2021harpo} propose Harpo, a reinforcement learning approach to obfuscate a user's browsing profile such that a subset of interest segments are kept while others are modified. Different from Harpo, \ToolX pairs the obfuscator with a denoiser to preserve the recommended videos related to the users’ actual interests while
removing the unrelated recommended videos caused by obfuscation. Moreover, the \ToolX obfuscator non-trivially adapts Harpo to YouTube, by building 1) a surrogate model with a different embedding model and loss function for replicating the YouTube recommendation system and 2) an obfuscator model which selects obfuscation videos based on the similarity between its embedding and the output embedding, such that it can handle an unlimited and varying number of possible obfuscation videos without requiring retraining.
Huang et al. \cite{huang2017context} propose a context-aware generative adversarial privacy (GAP) approach to train a ``privatizer'' for privacy-enhancing obfuscation against an {adversary} who attempts to infer sensitive information from input data. 
This approach is used to obfuscate mobile sensor data while navigating the privacy-utility tradeoff \cite{raval2019olympus,malekzadeh2019mobile,liu2019privacy}. While in theory we can apply GAP to jointly train the obfuscator and denoiser, in practice training them against YouTube in the wild which is prohibitively time consuming due to the iterative nature of GAP, and training them against the surrogate model is ineffective because the denoiser is able to trivially replicate the surrogate model 
(see Appendix \ref{appendix:joint}).
Beiga et al. \cite{biega2017privacy} propose a crowd-based obfuscation approach that allows individual users to preserve privacy by scrambling their browsing profiles via mediator accounts, which are selected such that the personalized recommendations to these mediator accounts are still coherent and  utility-preserving to the users behind each mediator account.
However, this approach requires a collaboration across multiple users of a recommendation system, and cannot be used by standalone users.

While recent work on privacy-enhancing obfuscation has attempted to balance the privacy-utility tradeoff, they are limited to obfuscating the input to the recommendation system to achieve this balance. 
These approaches are fundamentally limited as to how much utility can be preserved without undermining privacy by obfuscating the input to the recommendation system (see Fig. \ref{fig:budget}).
In contrast, \ToolX employs a two-step approach to this end. 
It first obfuscates the input to the recommendation system to preserve user privacy and then attempts to de-obfuscate the output recommendations to preserve utility. 
%

\section{Conclusion}
In this paper, we proposed \ToolX, a privacy-enhancing and utility-preserving obfuscation approach for YouTube's recommendation system that does not rely on cooperation from YouTube.

\ToolX used an \Obfuscator to inject obfuscation videos into a user's video watching history and a \Denoiser to remove the ``noisy'' recommended videos thus recovering the initial, unobfuscated recommendations. 
Our evaluation results demonstrated that \ToolX can reduce the utility loss by 2$\times$ for the same level of privacy compared to existing state-of-the-art obfuscation approaches.
Our work provides a template for implementing such utility-preserving obfuscation approaches on other similar online platforms, such as TikTok \cite{tiktok21} and Facebook \cite{FaceBook21}.
We will publicly release our code in conjunction with this paper to facilitate follow-up research.

\begin{acks}
The authors would like to thank Sean Hackett for his help
with the discussion of the denoiser idea, Muhammad Haroon for his help with the data collection and browser extension implementation, and Magdalena Wojcieszak for sharing the web browsing histories of real-world users.
This work is supported in part by the National Science Foundation under grant numbers 1956435, 1901488, and 2103439.
\end{acks}

\bibliographystyle{ACM-Reference-Format}
\bibliography{ref}

\appendix
\section{Why DP can not be guaranteed}
\label{appendix:dp}
\begin{theorem}
\label{theorem1}
Assume that there is one video, which the obfuscator ($O$, a randomized function) can not delete from a user persona ($P$), then we can not achieve $\epsilon$-DP or $(\epsilon,\delta)$-DP in terms of protecting the user persona.
\end{theorem}
\begin{proof}
\label{p1}
First, to achieve $\epsilon$-DP, for any two user personas $P_1$ and $P_2$ differing from one video, and for any user persona set $\mathcal{P}$ belonging to the output space of the obfuscator, $\frac{Pr(O(P_1)\in\mathcal{P})}{Pr(O(P_2)\in\mathcal{P})}\leq e^{\epsilon}$ should be satisfied. Now, assume that that there is one video $V$ which only exists in $P_2$ but not in $P_1$, and the obfuscator $O$ can not remove it from $P_2$ after obfuscation, which means $O(P_2)$ will always contain video $V$. Then, there exists an user persona set $\mathcal{P}$ which contains user personas without video $V$, where $Pr(O(P_1)\in\mathcal{P}) = 1$ but $Pr(O(P_2)\in\mathcal{P}) = 0$. Therefore, $\frac{Pr(O(P_1)\in\mathcal{P})}{Pr(O(P_2)\in\mathcal{P})}=+\infty$ and hence $\epsilon$ will be infinite in order to bound this worst case.

Second, to achieve $(\epsilon,\delta)$-DP, for any two user personas $P_1$ and $P_2$ differing from one video, and for any user persona set $\mathcal{P}$ belonging to the output space of the obfuscator, $|Pr(O(P_1)\in\mathcal{P}) - e^{\epsilon}Pr(O(P_2)\in\mathcal{P})|\leq \delta $ should be satisfied. Moreover, for $\delta$ to be meaningful, it has to be inversely proportional to the size of the dataset, which in our case is enormous (all possible user personas). However, since there exists an user persona set $\mathcal{P}$ without video $V$, where $Pr(O(P_1)\in\mathcal{P}) = 1$ but $Pr(O(P_2)\in\mathcal{P}) = 0$, the value of $\delta$ equals 1, which is meaningless in terms of $(\epsilon,\delta)$-DP.
\end{proof}

\begin{theorem}
Assume that there is one interest category, which the obfuscator ($O$, a randomized function) can not remove from a user profile (i.e. a list of interest categories) created by YouTube ($R$), then we can not achieve $\epsilon$-DP or $(\epsilon,\delta)$-DP in terms of protecting the user profiles.
\end{theorem}
\begin{proof}
Define the YouTube recommendation system as $R$. First, to achieve $\epsilon$-DP, for any two user profiles $R(P_1)$ and $R(P_2)$ differing from one interest category, and for any user profile set $\mathcal{R}$ in the output space of recommendation system, $\frac{Pr(R(O(P_1)\in\mathcal{R})}{Pr(R(O(P_2)\in\mathcal{R})}\leq e^{\epsilon}$ should be satisfied.
Now, assume that there is one interest category $I$ which is only in user profile $R(P_2)$ but not in user profile $R(P_1)$, and the obfuscator $O$ can not remove it from user profile $R(O(P_2))$, which means user profile $R(O(P_2))$ will always contain interest category $I$. Then, there exists a user profile set $\mathcal{R}$ containing user profiles without interest category $I$, where $Pr(R(O(P_1))\in\mathcal{R})=1$ but $Pr(R(O(P_2)\in\mathcal{R}) = 0$. Therefore, $\frac{Pr(R(O(P_1)\in\mathcal{R})}{Pr(R(O(P_2)\in\mathcal{R})}=+\infty$ and hence $\epsilon$ will be infinite in order to bound this worst case.

Second, to achieve $(\epsilon,\delta)$-DP, for any two user profiles $R(P_1)$ and $R(P_2)$ differing from one interest category, and for any user profile $R(P)$ in the output space of recommendation system, $|Pr(R(O(P_1))=R(P)) - e^{\epsilon}Pr(R(O(P_2))=R(P))|\leq \delta $ should be satisfied. However, since there exists a user profile set $\mathcal{R}$ containing user profiles without interest category $I$, where $Pr(R(O(P_1)\in\mathcal{R})=1$ but $Pr(R(O(P_2)\in\mathcal{R}) = 0$, the value of $\delta$ equals 1, which is meaningless in terms of $(\epsilon,\delta)$-DP.
\end{proof}

\section{System Design Details}

\subsection{MDP}
\label{appendix:mdp}
The obfuscation video selection process of \Obfuscator can be formulated as a Markov Decision Process (MDP) defined as follows:

\textit{1) Obfuscation step:} As shown in Figure \ref{fig:mdp}, at the beginning of each time step, a video will be played. If the played video is an obfuscation video injected by the \Obfuscator, we refer to this time step as an obfuscation step. We denote the number of videos that have been played up to obfuscation step $t$ by $n_t$. Note that we use the obfuscation budget $\alpha$ as a system parameter to control the percentile of obfuscation videos. At each time step, with probability $\alpha$, an obfuscation video will be injected by \Obfuscator into the user persona.

\textit{2) State:} We define state $s_t\in\mathcal{S}$ at obfuscation step $t$ as $s_t=[v_1,...,v_{n_t}]$, where $n_t$ is the total number of videos played until the beginning of obfuscation step $t$, and $\mathcal{S}$ is the state space of the MDP.

\textit{3) Action:} At obfuscation step $t$, an action $a_t$ will be taken by the MDP. We define action $a_t\in\mathcal{A}$ as the obfuscation video selected by the MDP policy, where $\mathcal{A}$ is the action space of the MDP, i.e. the obfuscation videos set in our application.

\textit{4) State Transition:} We define the state transition function as $\mathcal{T}(\cdot|\mathcal{S},\mathcal{A}):\mathcal{S}\times\mathcal{A}\times\mathcal{S}\rightarrow\mathbb{R}$, which outputs the probability of $s_{t+1}=s'$ given $s_t=s$ and $a_t=a$ as $\mathcal{T}(s_{t+1}=s'|s_t=s,a_t=a)$. In our system, state $s_{t+1}$ contains all videos played until state $s_t$, the action $a_t$ (i.e. the obfuscation videos selected at obfuscation step $t$), and all the videos played by users between obfuscation step $t$ and obfuscation step $t+1$. Note that the randomness of this MDP comes from the random injection of obfuscation videos.

\textit{5) Reward:} We associate a reward $r_t$ for the action $a_t$ at obfuscation step $t$. Specifically, we define $r_t$ as the difference of the privacy metric $P$ (see Eq. (\ref{eq:priavcy})) between this obfuscation step and the previous one, i.e., $r_t=P_t - P_{t-1}$, where $P_t$ represents the privacy metric value at obfuscation step $t$, calculated based on the recommended video class distributions of a non-obfuscated user persona and the corresponding obfuscated user persona at the end of obfuscation step $t$.

\textit{6) Policy:} The policy of the MDP can be defined as $\mathcal{\pi}(\cdot|\mathcal{S}): \mathcal{S} \times \mathcal{A} \rightarrow \mathbb{R}$, which outputs the probability of $a_t=a$ given $s_t=s$ as $\pi(a_t=a|s_t=s)$. In our system, the \Obfuscator is modeled as the policy in MDP, which outputs the probability distribution of obfuscation video selection. Suppose we have $M$ obfuscation videos in the obfuscation video set ($\mathcal{A}$), then we have $\sum_{i=1}^{i=M}\pi(a_t=i|s_t)=1$, where $a_t=i$ represents the selection of $i$-th obfuscation video. At each obfuscation step $t$, 
we randomly choose one obfuscation video based on a multinomial distribution parameterized by $A_t=[\pi(a_t=1|s_t),\cdots,\pi(a_t=M|s_t)]$, conditioning on the current state $s_t$. The goal of solving this MDP is to find the optimal policy, such that the accumulative rewards $\sum_{t=1}^{t=T}r_t$ can be maximized. Note that $T$ is the total number of obfuscation steps since we consider a finite-horizon MDP.
\begin{figure}[!h]
    \centering
    \includegraphics[width=.82\linewidth]{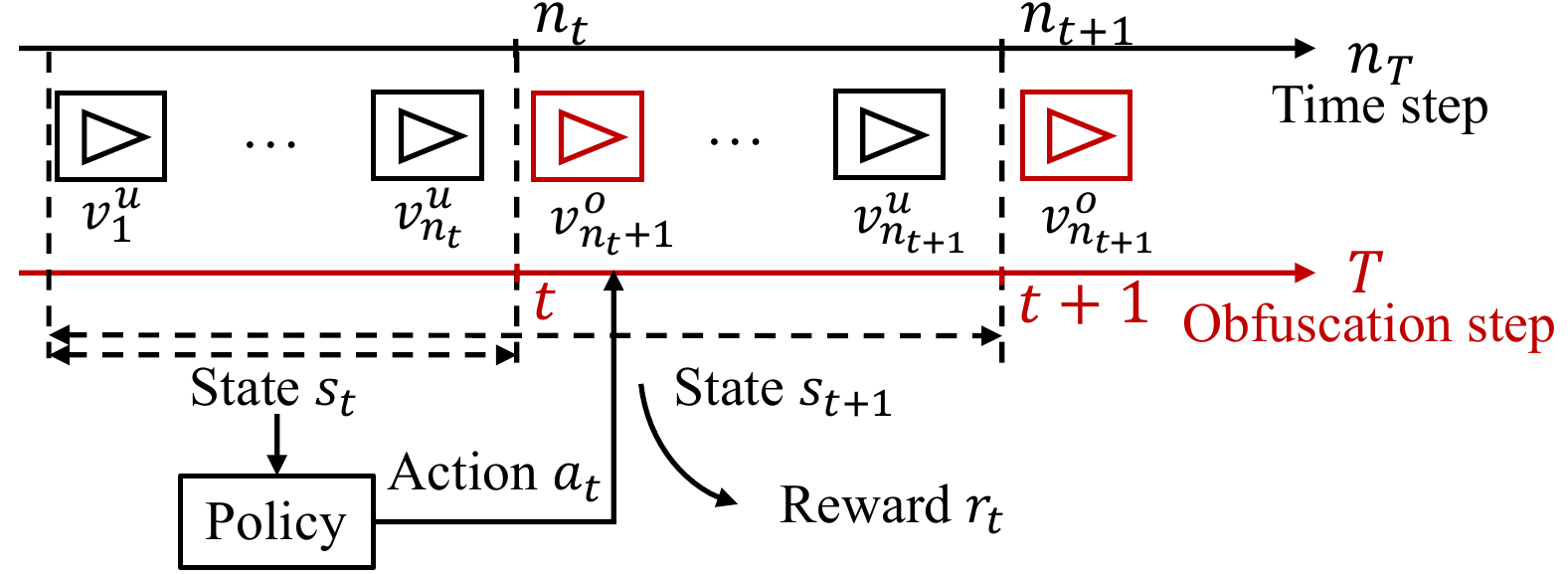}
    \vspace{-.1in}
    \caption{MDP for the \Obfuscator.}
    \vspace{-.15in}
    \label{fig:mdp}
\end{figure}

Note that the state $s_t$ (i.e. video sequence) will be continuously updated by appending new videos and is only growing unless users manually delete the history. \ToolX is designed to take the whole state $s_t$ as input of its \Obfuscator to select an obfuscation video, and then run denoising at each step. Hence, the calculation made by \ToolX at each step will depend on the calculation made by
\ToolX in the previous step, which is consistent with how YouTube works.

Moreover, we clarify that Harpo \cite{zhang2021harpo} and De-Harpo use a similar MDP formulation but with a different state (video sequence instead of webpage sequence) and reward function (privacy metric). They apply the same RL algorithm (A2C) to train the obfuscator, though the implementation differs due to MDP differences. 

\begin{figure}[!t]
\centering
\begin{subfigure}{.22\textwidth}
\centering
    \includegraphics[width=.99\linewidth]{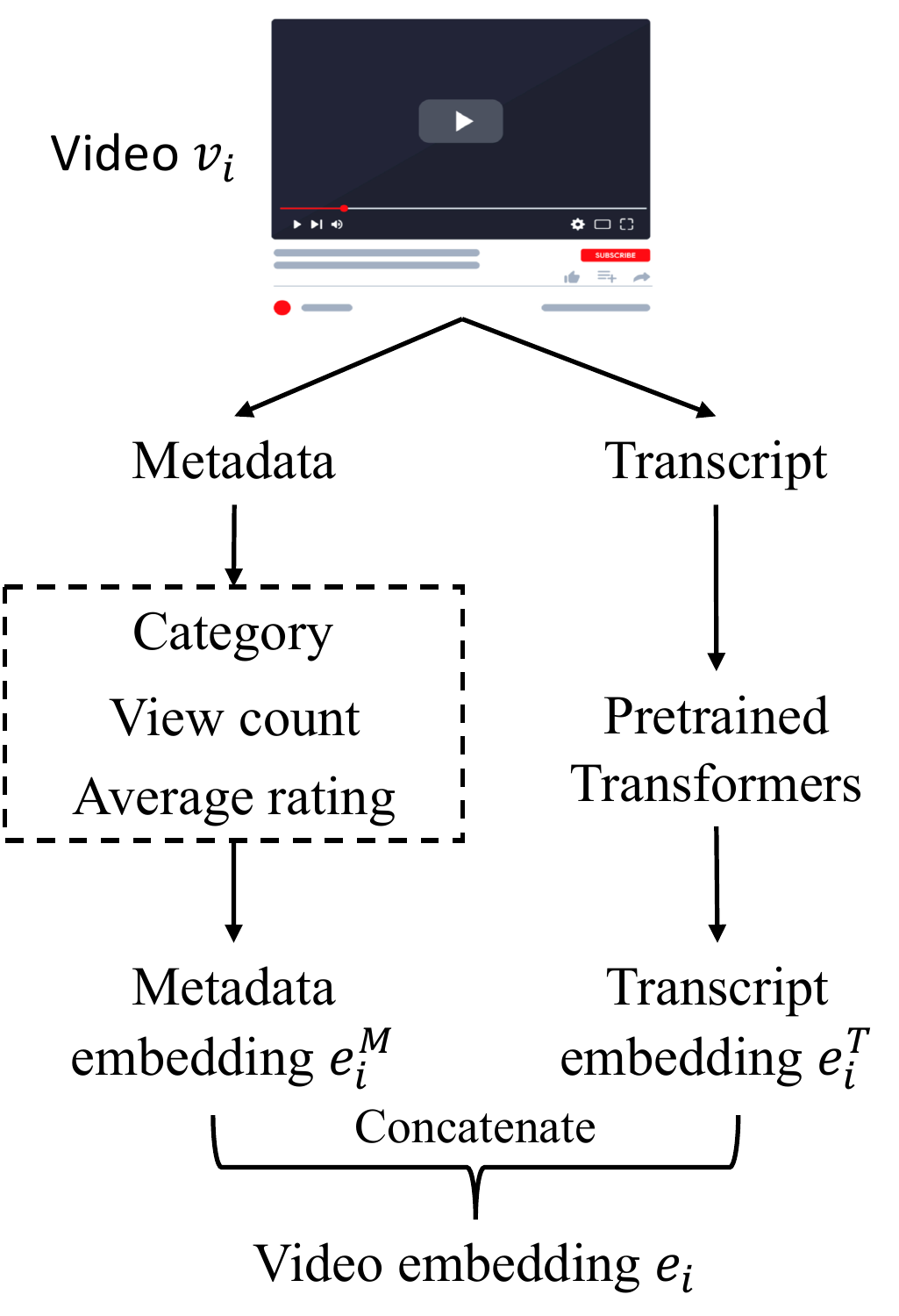}
    \caption{Video embedding}
    \label{fig:stru1}
\end{subfigure}
\begin{subfigure}{.22\textwidth}
\centering
    \includegraphics[width=.99\linewidth]{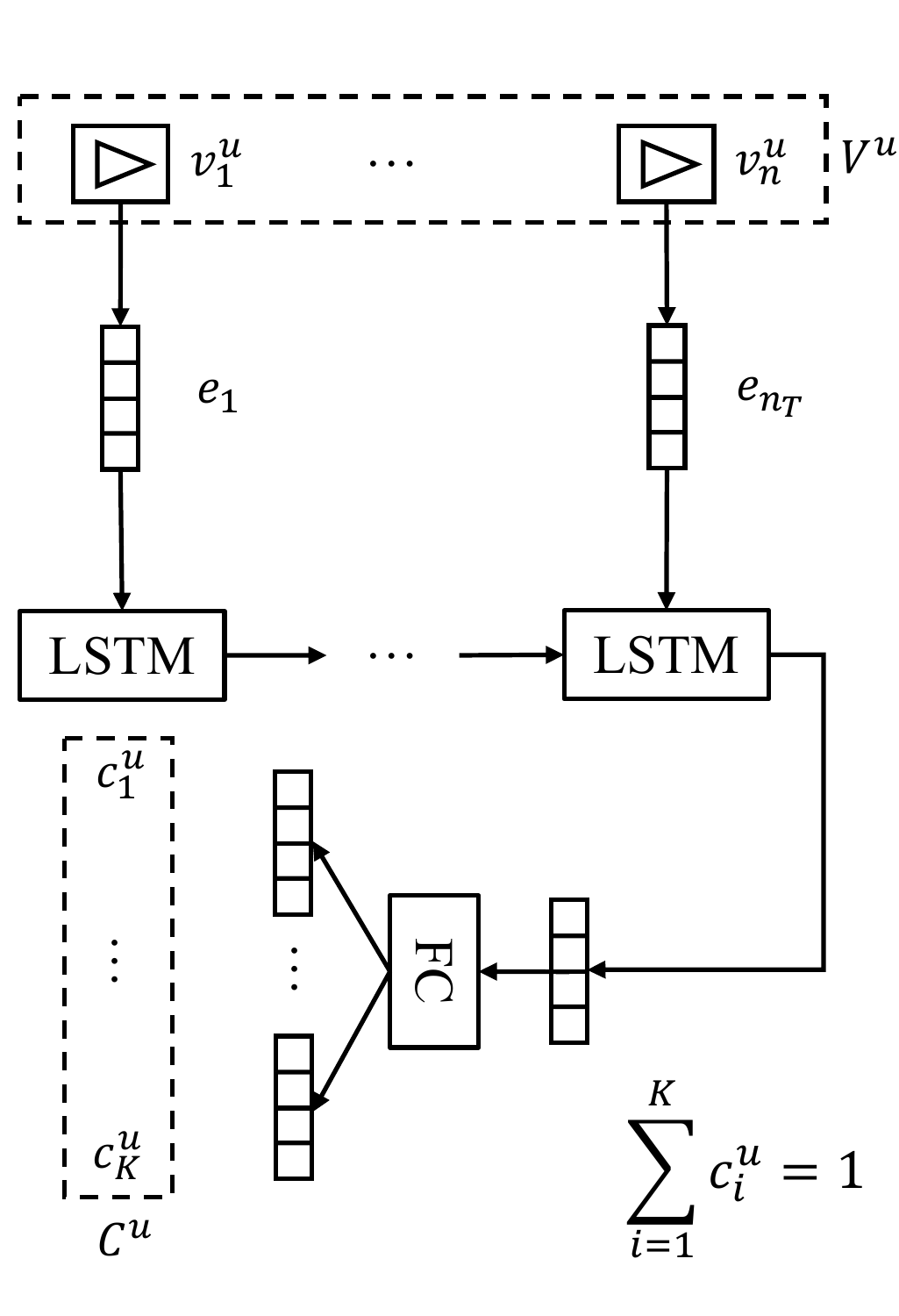}
    \caption{Surrogate model}
    \label{fig:stru4}
\end{subfigure}
\vspace{-.15in}
\caption{Details of system design. } 
\vspace{-.1in}
\label{fig:budget}
\end{figure}

\subsection{Video Embedding}
\label{appendix:video_emb}
As demonstrated in Figure \ref{fig:stru1}, we start by extracting the category, view count and average rating of each video from its metadata. We then use an one-hot embedding to represent the category of each video (with dimension 18)\footnote{Note that YouTube has 17 video categories, and we add an additional ``none'' category for videos without category metadata. Hence, the one-hot-embedding for category information has a dimension of 18.}, and use two real numbers to represent the standardized view count and average rating of each video. By combining them together, we derive the metadata embedding with 20 elements. We denote the metadata embedding for video $v_i$ as $e_i^M\in\mathbb{R}^{20}$.  

Next, we use a pretrained natural language processing (NLP) Transformer from \cite{URL_SENTENCE_TRANSFORMER} to generate the transcript embedding for the video transcript. Since the pretrained NLP Transformer has a constraint on the maximal number of words in the input text (256 words in our case), we firstly split video transcript with more than 256 words into multiple transcript chunks, each of which contains 256 words. Then, for each transcript chunk, we use it as input of the NLP model and get the output embedding vector. We take the average of these embedding vectors for these transcript chunks to derive the final transcript embedding. We denote the transcript embedding for video $v_i$ as $e_i^T\in\mathbb{R}^{384}$, which is a real vector with dimension 384. Note that if a video does not contain any transcript (e.g. music videos), we use the video title and description as an alternative of transcript to generate the transcript embedding. 
Last, we concatenate the metadata and transcript embeddings and derive the complete video embedding vector $e_i=[e_i^M,e_i^T]\in\mathbb{R}^{404}$.

\subsection{YouTube Surrogate Model}
\label{appendix:surro_model}
Prior approaches to learn latent user-item relationships for recommendation systems (e.g., matrix factorization \cite{mnih2007probabilistic,shan2010generalized,ma2008sorec,jamali2010matrix,ma2011recommender,tang2013exploiting,tang2016recommendations,yang2016social}, neural MF \cite{he2017neural,he2018outer,chen2020efficient,deng2016deep,wang2017item,zhao2017social,fan2018deep,fan2019deep}) are not scalable because they rely on a fixed set of users and items.
To address this limitation, recent work has focused on embedding based recommendation systems that predict the next item clicked by users from their item-click history and thus can scale to a large and dynamic set of users and items \cite{covington2016deep, chen2019top}.
YouTube, deals with a large influx of videos and users everyday \cite{URL_YOUTUBE_BLOG}
and thus uses a scalable recommendation system that predicts the next watched videos based on the embeddings of the past watched videos and other factors \cite{covington2016deep}.
Similar to YouTube's embedding based recommendation architecture, our surrogate model also takes as input the video embeddings. 
Slightly different from YouTube's embedding based recommendation architecture and as explained in Section \ref{subsubsec:rec_video_dist}, our surrogate model is designed to predict the recommended video class distribution, instead of making video-level recommendations. 

\section{Experimental setup}
\subsection{Sock Puppet Based Persona Model}
\label{appendix:sock_puppet}
Specifically, we denote this model as $G(D,T)$ parameterized by $D$, the depth of the recommendation trail, and $T$, the total number of videos in the watch history, and we define the \textit{recommendation trail} as a sequence of videos that are recommended and subsequently watched by a user starting from the given seed video.
At each step of the recommendation trail, we randomly select one ``up next" video to watch from the list of recommended videos with uniform probability.
We repeat this process until the \textit{recommendation trail} reaches the depth $D$ at which point we check if the user has watched $T$ videos.
If not, we randomly select another seed video from the user's homepage and repeat the process until $T$ videos have been watched. Note that we randomly select around 20,000 popular videos from a set of popular YouTube channels as our \textit{seed videos}. For each \textit{seed video}, we randomly generate a \textit{recommendation trail} and use it as a synthetic user persona.
Note that we collect seed videos used for generating sock puppet randomly from 200 popular YouTube channels, which include videos from all YouTube video categories.

\subsection{Data Preparation}
\label{appendix:dataset}
\paratitle{Video Embedding Preparation.} 
We use \texttt{youtube-dl}, a free software for downloading YouTube videos \cite{URL_YOUTUBE_DL}, to collect metadata and transcripts of videos. For metadata, we extract the category, average rating, view count, title, and description of each video, which is then used to generate the metadata embedding of each video. For a transcript, after we download it, we extract the transcript text, split it into text chunks with 256 words each, and use the pretrained Transformer \texttt{all-MiniLM-L6-v2} 
from \cite{URL_SENTENCE_TRANSFORMER} to convert them into transcript embeddings.  As described in Section \ref{subsubsec:video_emb}, we combine the metadata and transcript embeddings to generate the final video embedding.

\paratitle{User Persona Dataset Collected for Surrogate Model.}
We construct 10,000 sock puppet based personas and 10,000 Reddit user personas with 40 videos each.
For each of these personas, we collect the YouTube homepage recommended videos and derive the recommended video class distribution. We use these constructed personas as inputs ($V^u$) and the associated recommended video class distributions as labels ($C^u$) to build the dataset for surrogate model training and testing. As discussed in Section \ref{appendix:train_and_test}, we use supervised learning to train the surrogate model. 

\paratitle{User Persona Dataset Collected for Obfuscator and Denoiser.}
To evaluate the effectiveness of the \Obfuscator model against the real-world YouTube recommendation system, we need to construct both non-obfuscated and obfuscated user personas. Specifically, for each \Obfuscator model under an obfuscation budget $\alpha$, we first construct 2,936 non-obfuscated user personas \footnote{Note that 2,936 non-obfuscated user personas consist of 1,000 sock puppet based personas, 1,000 Reddit user personas, and the 936 real user personas from real-world users.}  with 40 videos each and the corresponding 2,936 obfuscated user personas generated by the \Obfuscator with on average $40*\frac{\alpha}{1-\alpha}$ videos each. Then for each pair of non-obfuscated and obfuscated user persona ($V^u$ and $V^o$), we collect their associated recommended videos from the YouTube homepage and derive their recommended video class distribution ($C^u$ and $C^o$).

Moreover, we use the same user persona data collected for the \Obfuscator evaluation to create the dataset for the \Denoiser training and testing (see Section \ref{appendix:train_and_test}). Specifically, each input of this dataset consists of one non-obfuscated user personas ($V^u$), the corresponding obfuscated user persona generated by the \Obfuscator ($V^o$), and its associated recommended video class distribution ($C^o$). Each label of this dataset is the recommended video class distribution of the non-obfuscated user persona ($C^u$).

\paratitle{Obfuscation Video Set.} 
We create our obfuscation video set by combining played videos during persona construction and videos appearing in homepage recommendations of all personas. In total, we collect approximately one million YouTube videos and use them as the obfuscation video set.
Note that the \Obfuscator will select one obfuscation video from the obfuscation video set at each obfuscation step.

\subsection{Ethical Issues Related to Reddit User and Real-world User Personas}

For the Reddit dataset, it is deemed exempt by IRB, and the dataset is publicly available and pre-crawled at https://files.pushshift.io/reddit/. We will de-identify usernames before public data release. For the YouTube users’ dataset we obtained an IRB approval and conducted experiments along the Menlo Report guidelines \cite{kenneally2012menlo}: Users consented to their data being collected for research purposes. We will not publicly release the dataset. 

\section{Training and Testing}
\label{appendix:train_and_test}

\paratitle{Surrogate Model.}
We split the user persona dataset collected for the surrogate model in 80\% for training and 20\% for testing. We use stochastic gradient descent for the surrogate model to minimize its loss, which is defined as the KL divergence between its output distribution and the actual recommended video category distribution of input user persona. We train our surrogate model for 50 epochs, where all the training samples are used once at each training epoch. We report that the average loss of our surrogate model on the testing dataset is 0.55.

\paratitle{Obfuscator.}
Recall that the \Obfuscator needs to take as input the non-obfuscated user personas. We use the training and testing user personas in the dataset collected for the surrogate model as the non-obfuscated user personas, and train the \Obfuscator to generate obfuscated user personas that maximize privacy (see Section \ref{subsec:sys_model}). Specifically, we train the \Obfuscator against the surrogate model for 50 epochs, where all the training user personas are used once at each epoch. After that, we use the testing user personas to evaluate the \Obfuscator against both the surrogate model and the real-world YouTube recommendation system, and report the average privacy metrics ($P$ and $P^{Norm}$). Note that to evaluate the performance of the \Obfuscator against YouTube, we construct non-obfuscated and obfuscated user personas to collect real-world data from YouTube (see Section \ref{subsec:dataset}). Moreover, when training the \Obfuscator, we use the on-policy RL algorithm A2C (Advantage Actor and Critic)\cite{openai2017a2c}, which is one of the state-of-the-art on-policy RL algorithms. Note that we choose the on-policy RL algorithm since it fits our application well, where the \Obfuscator (RL agent) needs to keep interacting with the YouTube recommendation system (environment) to improve the policy in an online fashion due to the dynamics of the YouTube recommendation system.

\paratitle{Denoiser.}
As described in Section \ref{subsec:dataset}, we create a dataset with 1,800 samples to train and test the \Denoiser, where 80\% of the samples are used for training and 20\% are used for testing. Specifically, 
the \Denoiser is trained via stochastic gradient descent to minimize the KL divergence between the output of the \Denoiser $\hat C^u$, i.e. the estimated recommendation video category distribution of a non-obfuscated user persona, and the actual distribution $C^u$. We train the \Denoiser for 50 epochs, where all the training samples are used once at each training epoch. We test the \Denoiser using the remaining 20\% samples and report the average utility metrics ($U_{Gain}^{Norm}$ and $U^{Loss}$).

Note that when we test \ToolX on sock puppet based persona dataset, we use the models of \ToolX trained on sock puppets dataset; when we test \ToolX on Reddit user persona dataset, we use the models of \ToolX trained on sock puppets dataset; when we test \ToolX on Real-world YouTube user dataset, we use the models of \ToolX trained on sock puppets dataset.

\section{System Overhead Analysis}
\label{appendix:sys_overhead}
We evaluate the system overhead of \ToolX in terms of CPU and memory usage and the video page load time using a an Intel i7 workstation with 64GB RAM on a campus WiFi network.
As described in Section \ref{subsubsec:sys_imp}, \ToolX consists of an \Obfuscator component that always runs in the background and a \Denoiser component that only runs when the user visits the YouTube homepage. We separately report their overhead below.

\textit{1) Obfuscator:}  We select an obfuscation budget $\alpha$ from $\{$0.0, 0.2, 0.3, 0.5$\}$, where $\alpha=0.0$ is used as the baseline (i.e. no obfuscation videos). 
For each obfuscation budget $\alpha$ we construct 10 user personas with 15 user videos each, and the browser extension visits $15 \cdot \alpha$ obfuscation videos in the background. 
We find that the increased CPU usage is less than 5\% and the increased memory usage is less than 2\%, even for obfuscation budget $\alpha=0.5$. 
Moreover, the change in video page load time of user videos is less than 2\% as $\alpha$ increases. 
Hence, we conclude that the \Obfuscator component in \ToolX has a negligible impact on the user experience overall.

\textit{2) Denoiser:} The YouTube's homepage load time with \ToolX is 1.79 seconds, which represents just a 37.8 millisecond increase as compared to the homepage load time without \ToolX. 
Specifically, it takes less than 24.6 millisecond to get the ``noisy" recommended videos from the homepage, 13.0 millisecond for the \textit{denoising} module to get ``clean" recommended videos, and 0.2 millisecond for showing these videos in the homepage.
In terms of the CPU and memory usage, the \Denoiser of \ToolX will increase them by 27.1\% and 2.2\% respectively, which is mainly due to running the ML model in the \textit{denoising} module.
Note that the increase of the CPU usage (from 12.9\% to 40.0\%) lasts for just 13 milliseconds while the ML model runs and
returns to the normal level right after that.
It is worth noting that the aforementioned measurements are conducted for the \textit{live} version of \ToolX.
In practice, we can reduce the overhead even further by implementing a \textit{cached} version of \ToolX, which caches the YouTube homepage periodically in the background and simply shows the cached homepage when the user navigates to the YouTube homepage. 
Hence, we conclude that the \Denoiser component in \ToolX has a negligible impact on the user experience overall.

\section{Discussion of Joint Training of Obfuscator and Denoiser}
\label{appendix:joint}
The \Obfuscator and \Denoiser in \ToolX are separately trained and their joint training might be much more effective. 
We experimented with jointly training the \Obfuscator and the \Denoiser using multi-objective reinforcement learning. 
Specifically, we started by training a \Denoiser model. 
Then, we trained the \Obfuscator to maximize the privacy against the surrogate model, while minimizing the loss of the \Denoiser with obfuscated user personas as inputs. 
After we trained the \Obfuscator, we retrained the \Denoiser and repeat the above process until both the \Obfuscator and the \Denoiser converge. 
We found that jointly training did not improve privacy or utility because of our use of the surrogate model, instead of YouTube in the wild, for practical reasons. 
When trained against the surrogate model, \Denoiser was able to trivially replicate the surrogate model. 
While in theory we could jointly train the \Obfuscator and the \Denoiser in the wild to avoid this issue, it would not be practical due to its time consuming nature. 
Future work can look into hybrid surrogate and in the wild joint training of  \Obfuscator and \Denoiser.

\end{document}